%% file: arxivdec19.tex
\documentclass[11pt,a4paper]{article}
\usepackage[margin=1.0 in]{geometry}
\usepackage[utf8]{inputenc}
\usepackage{amsfonts}
\usepackage{hyperref}
\usepackage{amsthm}
\usepackage{amssymb}
\newcommand{\ignore}[1]{}
\newtheorem{lemma}{Lemma}
\newtheorem{theorem}{Theorem}
\newtheorem*{theorem*}{Theorem}
\newtheorem{definition}{Definition}

\newtheorem{remark}{Remark}
\usepackage{amsmath,color}
\usepackage[ruled,linesnumbered,vlined]{algorithm2e}
\usepackage[usenames,dvipsnames]{xcolor}
\usepackage{textcomp}
\usepackage{amssymb}
\usepackage{cite}
\usepackage{graphicx}
\usepackage{url}

\usepackage[affil-it]{authblk}

\title{Coalition-Safe Equilibria with Virtual Payoffs}
\author{Aggelos Kiayias %
  \thanks{Electronic address: \texttt{akiayias@inf.ed.ac.uk}}}
\affil{ University of Edinburgh 
\& IOHK}

\author{Aikaterini-Panagiota Stouka%
  \thanks{Electronic address: \texttt{a.stouka@ed.ac.uk}  
  }}
\affil{University of Edinburgh \& IOHK} 

\begin{document}

\maketitle

\input{Abstractarxiv}

\input{Introarxiv}

\input{body}

\bibliography{equilibriumfull}
\bibliographystyle{plain}

\appendix
\input{appendix}

\end{document}

%% file: Abstractarxiv.tex
\begin{abstract}
Consider a set of parties invited to execute a protocol $\Pi$. The protocol  will incur some cost to run while in the end (or at regular intervals), it will populate and update local tables that assign   (virtual) rewards to participants. Each participant aspires to offset the costs of participation by these virtual payoffs that are provided in the course of the protocol. 
In this setting, we introduce and study a notion of coalition-safe equilibrium. 
In particular, we consider a  strategic coalition of participants that is centrally
coordinated and potentially deviates from $\Pi$ with the objective to
increase its utility with respect to the view of {\em at least one} of the other participants.  The protocol $\Pi$ is called a coalition-safe equilibrium with virtual payoffs (EVP) if no such protocol deviation exists. 
We apply our notion to study incentives in blockchain protocols. Compared to prior work, our framework has the advantages that it simultaneously (i)  takes into account that each participant may have a divergent view of the rewards given to the other participants, as the reward mechanism employed is subject to consensus among players (and our notion is well defined independently of whether the underlying protocol achieves consensus or not)
(ii) accounts for the stochastic nature of these protocols enforcing the equilibrium condition to hold with overwhelming probability (iii) provides a versatile way to describe a wide variety of utility functions that are based on rewards recorded in the ledger and cost incurred during ledger maintenance. We proceed to use our framework to 
  provide a unified picture of incentives in the Bitcoin blockchain,  for absolute and relative rewards based utility functions, as well as  prove novel results regarding incentives of the Fruitchain blockchain protocol [PODC 2017]  showing that the equilibrium condition holds for collusions up to $n-1$ players for absolute rewards based utility functions and less than $n/2$ for relative rewards based utility functions, with the latter result holding for any ``weakly fair'' blockchain protocol, a new property that we introduce and may be of independent interest. 
\end{abstract}

%% file: Introarxiv.tex
\section{Introduction}

A game involves a number of  participants that engage with each other following a certain strategy profile which incurs individual costs and rewards. The utility of each participant, which rational participants aspire to maximize,  is some compound real-valued function that takes into account the costs incurred and rewards  resulting by the interaction. A common characteristic is that costs and rewards are bestowed authoritatively via some infrastructure that is typically external to the game execution. 
Contrary to this, in this work, we study a game-theoretic setting where rewards are {\em virtual} and are recorded as an outcome of the interaction of the participants individually in each participant's local view. Thus, while costs are incurred authoritatively as before, rewards are ``in the eye of the beholder'' and in the end of the interaction two participants may have diverging views about the rewards that each game participant has received, while any single participant $P$ cares fundamentally that the other participants conclude in their local views that $P$ has received rewards. 

Our motivation comes from the setting of distributed ledgers. These protocols were originally studied as an instance of the state machine replication problem~\cite{DBLP:journals/csur/Schneider90} but recently were 
popularised again due to the introduction of the Bitcoin blockchain protocol~\cite{nakamoto}. Bitcoin is a cryptocurrency based on a blockchain protocol that maintains a public ledger containing the history  of all transactions. The protocol was formally analyzed in the cryptographic setting in \cite{backbone,rafael}. The main idea behind the protocol is that transactions are organized into blocks and blocks form a chain, as each block contains the hash of  the previous block. The longest chain is selected to determine the public ledger. A block is produced when a proof of work puzzle \cite{puzzle,puzzle2,puzzle3,puzzle4} is solved by a node called  miner. The miner that produces a block earns an amount of Bitcoin as a reward.
One distinguishing feature of blockchain protocols is the emphasis they put on the incentives of the participating entities. Classically, consensus \cite{DBLP:journals/toplas/LamportSP82} was considered in various threat models, such as fail-stop failures or Byzantine. However the incentive and game theoretic
aspects of the protocol have received less attention.


In blockchain protocols, the rewards that are bestowed to the participants are not assigned in an authoritative manner by some external entity, but rather are recorded as an outcome of bookkeeping that takes place by the interaction of the participants. In such setting, the relevant question is 
whether a strategic coalition of participants has an incentive to follow the protocol or to deviate. In its simplest form we consider a ``monolithic'' such coalition (abstracted as an adversary) that considers deviating  from the protocol in a coordinated fashion with the aim to increase the joint utility of its members. 

 Different aspects of incentives in Bitcoin were studied in \cite{demystifying,kroll,selfish,mininggame,selfish3,baloon,whale,bribery,zikas}  and some type of incentive compatibility for
blockchain protocols was studied in the context of a few  protocols, see e.g.,  \cite{tor,fruitchain,sol} 
(cf.  Appendix \ref{definitions} for background information on game theoretic notions). 
With respect to studying the participation in the core blockchain protocol, 
Kroll et al. in~\cite{ kroll} show that  a certain modeling of the Bitcoin protocol is a Nash equilibrium, while
Eyal and Sirer in~\cite{selfish} show that Bitcoin is not incentive compatible because of a type of attack called selfish 
mining that works for any level of hashing power (for \textit{Nash equilibrium} and \textit{incentive compatibility} definition see Appendix \ref{definitions}). Then again, Kiayias et al. in~\cite{mininggame}  show that there are thresholds 
of hashing power where certain games that abstract Bitcoin have honest behavior as a 
Nash equilibrium. The above seemingly contradictory results stem from differences in the game theoretic modeling of the underlying blockchain protocol and the utility function that is postulated. In addition, the existing notions of equilibria (cf. Section~\ref{sec:otherrelated} below) do not appear to be sufficient to completely capture the rational behavior of participants. First, given the anticipated
long term execution length of such games it is important to consider the variance of utility and thus merely looking at expected utility might be insufficient. Second, 
the reward mechanism employed is subject to consensus
among participants and given that the protocol itself aims to achieve such consensus, each participant may have a divergent view 
of the rewards given to the other participants. Thus it is important
that the model used to examine the protocol takes into account the possibility of such
divergence and  the game should be well defined independently of whether the resulting interaction achieves consistency or safety, as such properties should be the result of the rational interaction of participants, not a precondition for it! 


%

\subsection{Our Results}

\paragraph{Execution model:}   Our model \ignore{strictly} generalizes the  execution model of \cite{backbone} and it is based  on the ``real-world'' protocol execution model of   \cite{Canetti2,Canetti1,UC2,Canetti4} with the additional feature that certain operations of the protocol are abstracted as oracles and calling such oracles incurs a certain cost to the callee. In this way, 
the cost of each participant is solely dependent on  participants' actions and aggregates the expenditure that is  incurred during the execution based on the oracle queries posed. For example in the case of a 
proof-of-work blockchain protocol this may amount to the number of queries posed to the hash function.

\paragraph{Utility with Virtual Payoffs:} At any point  of an execution, each participant has a local view regarding the virtual rewards of all participants, including themselves. 
The key observation for defining utility in our setting is that given that the rewards are virtual, it is not particularly advantageous for  a participant to be in a state where  
according to its own bookkeeping she has collected some rewards; instead what is important, is what {\em other} participants believe about one participant's rewards.
In this way we define two types of reward functions $R^{\max},R^{\min}$ which will correspondingly give rise to two utility functions. The $R^{\max}$ rewards of a coalition represent the maximum amount of rewards a coalition has received quantified over all {\em other} participants (which do not belong to the coalition), while $R^{\min}$ is similarly the minimum amount of rewards. 

\paragraph{Equilibria with Virtual Payoffs (EVP):}
Based on these functions (reward, cost and utility functions), we present a formal notion   of approximate Nash equilibrium, called coalition-safe Equilibrium with Virtual Payoffs (EVP). 
Informally, a protocol $\Pi$ is an EVP if it   guarantees that with    overwhelming probability, a  rational strategic actor (hence called the {\em adversary}) who controls a coalition of participants, cannot gain by deviating more than an insignificant amount in terms of utility in the view of {\em any} of the other participants.
As a result, for a given protocol $\Pi$,  if there is  a small, but non-negligible,  probability that the utility of the adversary deviating from $\Pi$  becomes significantly higher in the view of a single other participant then such protocol will {\em not be} an EVP. 
  
In more details,  our notion of equilibrium is defined by  examining two independent executions of the protocol in question. In the first execution the adversary controlling a coalition follows the protocol while in the second execution it might deviate in some strategic fashion. In both executions the participants who are not controlled by the adversary (we refer to them as honest participants) follow the protocol.  The way in which we examine these two executions is by comparing the utilities of the adversary in these two executions for all possible environments. The underlying protocol is EVP when with overwhelming probability the $U^{\max}$ utility of the adversary  when it deviates is not significantly higher compared to its $U^{\min}$ utility when it follows the protocol. 
This means that in order for our protocol not to be an EVP, there will be an alternative  strategy 
and an environment with respect to which,  
  the execution where the adversary deviates in the view of one honest participant results,  with a non-negligible probability,  to a significantly higher utility 
  compared to the {\em lowest} utility determined for the adversarial coalition  in the execution where it follows the protocol when quantified over {\em all} the honest participants.

\paragraph{EVP Analysis of Blockchain Protocols:}
In our analysis,  we revisit three important utility definitions for blockchain protocols: 
(i) absolute rewards (ii) absolute rewards minus absolute cost and (iii) relative rewards. With the term absolute rewards  we refer to  the amount of the rewards that a set of participants 
 receives at the end of the execution. With the term absolute cost we mean the cost that this set of participants pays during the execution expressed in absolute
terms.  With the term relative rewards we refer to the rewards of this set of participants divided by the total rewards given to all the participants.  
We note that the first and the third type of utility have been considered in a number of previous works, specifically, 
\cite{kroll,fruitchain} used the first type and
\cite{mininggame,tor,selfish}  used the third type. In addition the second type was used in \cite{absolutecost,zikas,gap}.
\ignore{With the term absolute rewards  we refer to  the amount of the rewards that a set of participants 
 receives at the end of the execution. With the term absolute cost we mean the cost that this set of participants pays during the execution expressed in absolute
terms.  With the term relative rewards we refer to the rewards of this set of participants divided by the total rewards given to all the participants.  
We note that the first and the third type of utility have been considered in a number of previous works, specifically, 
\cite{kroll,fruitchain} used the first type and
\cite{mininggame,tor,selfish}  used the third type. In addition the second type was used in \cite{absolutecost,zikas,gap}... .XXX describe them - and why XXX and who else has considered them XXX. 
'' This type of cost model is consistent with cloud mining~\cite{cloud} where participants establish a contract that may last a long period of time and they pay a fixed rental fee per time unit.} 

Using our model we prove positive and negative results regarding the incentives in Bitcoin unifying previous seemingly divergent views on how the protocol operates in terms of incentives, cf. Theorems \ref{absolutebitcoin},\ref{absolutecostbitcoin},\ref{relativebitcoin}. Specifically, we prove that Bitcoin with fixed target is an EVP in the static setting with utility based on absolute rewards, and absolute rewards minus absolute costs, while it is not with respect to relative rewards, cf. Figure~\ref{fig:table}. 

Next, we prove regarding incentives of Fruitchain, \cite{fruitchain}, the following new result: when the utility is based on  absolute rewards minus absolute cost, the Fruitchain protocol is an EVP in the static  setting against a coalition including even up to {\em all but one} of the participants (Theorem \ref {absolutefruitchain}). Moreover we define a property called ``$(t,\delta)$-weak fairness'' that is weaker than ``fairness'' defined in \cite{fruitchain} or ``ideal chain quality'' described in \cite{backbone} and the ``race-free property'' in \cite{tor} (for more details see Section \ref{fruitfair}) and is sufficient for proving that a protocol is EVP when the utility is based on  relative rewards (Theorem \ref{theorem8}). This allows us to also prove the following result: when the utility is based on relative rewards, the Fruitchain protocol is EVP in the static synchronous setting against any coalition including fewer than half of the number of the participants  (assuming participants of equal hashing power, cf. Theorem \ref{relativefruitchain}). Further, we note that  the approximation factor in the EVP is merely a constant additive factor. 
Regarding the level of rewards, in \cite{fruitchain} the total rewards $V$ of an execution are derived from from (a) the flat rewards of the fruits (for details regarding what a fruit according to \cite{fruitchain} is, see subsection \ref{fruitabstract}) and (b) the transaction fees from the transactions inside the fruits; in both cases these are distributed evenly among the miners and  $V$ is a fixed constant in the whole execution. 
Our result is also stronger in this respect, for both absolute and relative rewards based utilities, where we show that the protocol is an EVP even if rewards are a function of the security parameter or the length of the execution. 
\ignore{
 Instead, (i) when we study what happens in the Fruitchain protocol when the utility is equivalent to relative rewards we consider, like \cite{fruitchain}, that each fruit does not give the same reward (for example because of the transaction fees from the transactions it includes) and at the end of the execution the total rewards are distributed evenly among the miners of the fruits (but without the assumption that $V$ is a fixed constant) and (ii) when we study what happens in the Fruitchain when the utility is equivalent to absolute rewards minus absolute cost we assume that each fruit gives to its owner a fixed amount $w_f$ of money. The reason is that the assumption that $V$ is fixed does not always hold when the difficulty in mining a block is fixed. For example the adversary can reduce the total rewards and transaction fees by not asking all its queries, or by withholding honest blocks. We leave for future work to study what happens in \cite{fruitchain} when utility is equivalent to absolute rewards minus absolute cost, each fruit does not give the same amount of money and the total rewards  $V$ are shared evenly among the miners (without the assumption used in \cite{fruitchain} that $V$ is a fixed constant). 
 }
 
We note that our model is synchronous and in our results we consider that the adversary is static and decides in the beginning of the execution the participants it will control and the cost it will pay during each round.  We will refer to it  as ``static adversary with fixed cost.'' This type of cost model is consistent with 
cloud mining~\cite{cloud} where participants establish a contract and they pay a fixed rental fee per time unit. In addition we suppose that the difficulty in mining a block is fixed.  Interesting directions for future work is devising protocols that are EVPs against  a dynamic adversary which adaptively fluctuates its mining resources, while the protocol itself adjusts mining difficulty; designing and proving that  such EVP protocols exist is an interesting open question. 

\begin{figure}[h]
\begin{center}
\begin{tabular}{|p{4cm}|c|c|c}
\hline
                     & AbsR/AbsR-C & RelR \\
\hline
Bitcoin fixed target & $n-1$ $^{(\ast)}$  & NO$^{(1)}$  \\
Bitcoin variable target & NO$^{(2)}$ & NO$^{(3)}$  \\
Fruitchain & $(n-1)$ $^{(\dagger)}$ &  $<n/2$ \\
\hline
\end{tabular}
\end{center}
\caption{\label{fig:table}Overview of our results as well as previous results that are consistent with the EVP model. AbsR stands for a utility based on absolute rewards, AbsR-C for a utility based on  absolute rewards minus absolute cost, while RelR stands for a utility based on relative rewards. The function in $n$ specifies the larger coalition size for which the equilibrium stands. $^{(1),(3)}$ are
derived from \cite{selfish}, 
$^{ (2)}$ is derived from \cite{koutsoupias}.  The $^{(\ast)}$ result is informally postulated in \cite{kroll}. A weaker bound of the $^{(\dagger)}$ result in terms of coalition size $(<n/2)$ was shown in \cite{fruitchain}. }
\end{figure}
 
\subsection{Other Related Work}
\label{sec:otherrelated}

\ignore{
Our model is a generalization of the execution model of the Bitcoin Backbone protocol~\cite{backbone} and it is based on \cite{Canetti2,Canetti1,UC2,Canetti4}. Some important differences are: (i) in our model we consider that each query to the random oracle that models a cryptographic hash function has a cost and (ii) our model includes more than one oracles that may model different cryptographic primitives. \par 
Note that in our model there exist honest participants and a rational adversary that controls 
a number of protocol participants and tries to maximize their joint utility. This reflects the Nash equilibrium notion for coalitions where we examine if a coalition
of rational participants has incentives to follow the indicated strategy given that the other participants follow their indicated strategy.} 

A closely related work that focused on Byzantine Agreement and rational behavior is   \cite{10.1007/978-3-642-31585-5_50}. 
Some distinctions between our work and \cite{10.1007/978-3-642-31585-5_50}  are that (i)
their utility model is tailored to the setting of  (single shot) binary Byzantine agreement, 
while we focus on distributed ledgers that record transaction and rewards for the participants, 
(ii) in the definition of equilibrium they  consider the expectation of utility as opposed
to bounds on utility that are supposed to hold with high probability,
(iii) at equilibrium, the rational adversary may deviate from the protocol 
as long as the properties of Byzantine agreement are not violated, while we consider any
protocol deviation as potentially invalidating our equilibrium objective as long as the adversarial
coalition
benefits in the view of one of the other participants.

One model,  introduced in 
\cite{Aiyer:2005:BFT:1095810.1095816},  that combines Byzantine participants, i.e., participants that can deviate from the protocol arbitrarily, in addition to honest and rational participants, is ``BAR.'' This model   includes three types of participants: altruistic,
Byzantine and rational  and was used to analyse two types of protocols, IC-BFT (Incentive-Compatible Byzantine Fault Tolerant) and Byzantine Altruistic Rational 
Tolerant (BART) protocols \cite{Aiyer:2005:BFT:1095810.1095816}. The first type of protocols (i) satisfies the security properties of a Byzantine Fault Tolerant protocol (safety and liveness) in a setting with Byzantine/honest participants and (ii) guarantees that the best choice for rational participants is to follow the protocol. This guarantee is provided under the following assumptions: 
(a) if following the protocol is a Nash equilibrium then the rational participants will adopt it as a  strategy,  (b) rational participants do not collude,
and (c) the expected utility of the rational participants is computed considering that the Byzantine participants react in such a way that minimizes the utility of the rational participants. One of the advantages of the IC-BFT model is that it can be used to argue that  rational participants have incentives to follow the protocol due to property (ii) and thus they can be considered as honest and in such case the resulting protocol will still be resilient to some Byzantine behaviour due to property (i). \par Another game theoretic notion that takes into account malicious and rational participants in the context of multi-party computation is called  ``$\epsilon$-$(k,t)$-robust Nash equilibrium'' defined in \cite{defin}.
%
In this type of equilibrium no participant in a coalition of up to $k$ participants should be able to increase their utility given that there exist up to $t$ malicious participants. Note that in our case following \cite{fruitchain} when we consider coalitions we study their joint utility (by summing individual rewards) and not the utility of each participant separately something that results in a more relaxed notion in this respect (but still suitable for the distributed ledger setting: following \cite{backbone,fruitchain} when we study proof of work cryptocurrencies, each participant represents a specific amount of computational power. So a coalition of participants could also be thought to represent one miner). 

In \cite{rationalframework} a framework for ``rational protocol design'' is described that is based on the  simulation paradigm. 
That framework was extended and used for examining the incentive compatibility of Bitcoin in \cite{zikas}. The basic premise  is that the miners aim to maximize their expected revenue and the framework describes a game between two participants: a protocol designer D and an attacker A. The Designer D aims to design a protocol that maximizes the expected revenue of the non adversarial participants and keep the blockchain consistent without forks. The adversary A aims to maximize its expected revenue. \ignore{In \cite{zikas} the assumption is that the miners aim to maximize their expected revenue and the framework describes a game between two participants: a protocol designer D and and an attacker A. The Designer D has target to design a protocol that maximizes the expected revenue of the non adversarial participants and keeps the blockchain consistent without forks. The adversary A has target to maximize its expected revenue.}
One difference of our model compared to \cite{zikas}  is that we let the adversary deviate from the protocol not only if its expected utility increases significantly by deviating, but even if it can increase its actual utility significantly just with not negligible probability. In addition \cite{zikas} focuses exclusively on the incentive compatibility of  Bitcoin  and only when utility is equivalent to absolute rewards minus absolute cost. 

Other related works that study the incentive compatibility of Bitcoin according to a specific utility are \cite{kroll,mininggame,selfish}.  In addition, the incentives of nodes who do not want necessarily to engage in  mining but they want to use the Bitcoin system for transactions have been studied in \cite{users}. \ignore{ Note that in  \cite{mininggame} there exists also a reference to 
Individual Rationality, which means that the expected utility is not negative. This is also related to the participation constraint used in optimization problems in economics \cite{brit} and has appeared also in \cite{zikas}. We have taken similar considerations into 
account in order to define the notion of an incentive compatible blockchain protocol.}\ignore{
Regarding the notion ``incentive compatibility'' we use :
(i)``Dominant-strategy incentive-compatibility'' refers to a setting where there is no  better strategy  than telling the truth or following the protocol irrespectively of what  the other participants do. ii)``Bayesian-Nash incentive-compatibility''  is a  weaker notion when a participant has incentives to tell the truth given  that all the other participants do the same, cf. \cite{multi}.}\par  
As we already mentioned, in \cite{fruitchain}  the  Fruitchain protocol is presented, which preserves the security properties of Bitcoin protocol and satisfies a $\delta$-approximate fairness property (assuming honest majority) that is shown to be enough for incentive compatibility when the utility is equivalent to absolute rewards. In addition, in \cite{fruitchain} a definition of approximate  Nash equilibrium is described, denoted by``$\rho$-coalition-safe $\epsilon$-Nash equilibrium'' that guarantees protocol conformity with overwhelming probability.
Our EVP definition  is both more general and more explicit in the sense that: (i) It includes a formal description of the properties of the protocol's executions that give rise to the random variables that should be compared. (ii) It includes a formal definition of reward and utility functions. (iii)  It takes into account in a rigorous way the fact that local views of honest participants may diverge and it is  well defined even when  the underlying protocol view of participants are inconsistent.  
\ignore{
As mentioned above, our results are stronger in the sense that (i) we prove that the Fruitchain protocol is an EVP when the utility is equivalent to absolute rewards against a coalition including even up to all but one participants (not fewer than half of the participants as in \cite{fruitchain}) 
 (ii) we prove that Fruitchain is EVP also when the utility is equivalent to absolute rewards minus absolute cost (iii) we prove what happens in \cite{fruitchain} when the utility is equivalent to relative rewards showing in this way formally what  is the actual advantage of the Fruitchain over Bitcoin (because when utility is equivalent to absolute rewards and the difficulty in mining a block is fixed both Bitcoin and Fruitchain are EVP).}
\ignore{When the adversary wants to maximize the absolute rewards and the target of the blocks is fixed the results regarding incentive compatibility are almost the same in Bitcoin and Fruitchain and hold, as we have proved, without honest majority assumption. In addition Fruitchain protocol does not examine what happens when each query to the random oracle has a cost which cannot be neglected as mining needs a great amount of electricity.} 
\par
\ignore{
As we have already explained we use in one of our proofs a notion of ``weak fairness'' that we define and is weaker than ``fairness'' described in \cite{fruitchain}, ``ideal chain quality'' described in \cite{backbone} and ``race-free property'' in \cite{tor} Another property related to ``fairness'' is ``t-immunity'' in \cite{defin}. The last property considers utility as an expectation. Note that the notion of fairness has also been  used in \cite{inclusive}. A notion of weak fairness has also 
been used in \cite{fairness} for a different purpose. Specifically in \cite{fairness} fairness refers to exchanges between participants; both or neither of the participants take the other's item.} Some other works that investigate the interplay between Cryptography and Game theory in different settings are  \cite{costlycom,gamecrypto,gc,defin,faircomputation}. \ignore{In \cite{defin, faircomputation} there are also some definitions related to Nash equilibrium that refer also to strategic coalitions, but the utility used is also based on expectation.}Some proof-of-stake blockchain protocols (protocols that do not rely on proof of work to achieve consensus) that can be proved to be incentive compatible using some notion of equilibrium are \cite{ouroboros,snow}. A framework for identifying attacks against the incentive schemes of the blockchain protocols is proposed in \cite{squir}. In \cite{latest}, proof of work blockchain protocols are modeled as stochastic games while in \cite{survey} a survey of game theoretic applications in the blockchain setting is presented. \par Previous works on the general topic of rational multi-party protocols include 
\cite{10.1007/978-3-319-02786-9_14,Abraham:2008:ATP:1400751.1400804,DBLP:conf/podc/DaniMRS11,DBLP:conf/tcc/FuchsbauerKN10,10.1007/11818175_11}
while a related line of research explored cheap talk
\cite{defin,DBLP:journals/geb/UrbanoV04,DBLP:conf/podc/LepinskiMP04,cheap3}.
For example cheap talk \cite{cheap1,cheap2} was used in   \cite{defin} for simulating  an honest mediator given (i) secure private channels between agents that incur no cost, (ii) a punishment strategy such as having the participants stop the protocol if misbehaviour is detected. \par 
A game theoretic notion that can be used to handle protocols operating in 
asynchronous networks is the ``ex post Nash equilibrium" and was used in this context in \cite{10.1007/978-3-642-41527-2_5,Halpern:2016:RCE:2933057.2933088}. The way this was used in our context, was to include also adversarial nodes in addition to rational nodes and in \cite{10.1007/978-3-642-41527-2_5} the adversarial nodes would determine some specific choices in the protocol execution 
(such as the initial signal the agents get and the order in which agents are scheduled). The  equilibrium condition is required to hold regardless of the choices of the adversarial nodes and even if the rational participants know these choices.\par 

\ignore{For related work regarding other types of multi-party protocols, see ``ex post Nash equilibrium'' and ``cheap talk'' \iffalse \cite{defin,DBLP:journals/geb/UrbanoV04,DBLP:conf/podc/LepinskiMP04,cheap3} \fi in Appendix \ref{related}.} 

Another property (apart from these we have already referred to) related to ``fairness" is ``t-immunity" in \cite{defin}. This property also considers utility as an expectation. Note that the notion of fairness has also been  used in \cite{inclusive}. A notion of weak fairness has also 
been used in \cite{fairness} for a different purpose. Specifically in \cite{fairness} fairness refers to exchanges between participants; both or neither of the participants receive the intended output. \par Finally we note that coalition-safety has been examined also in the context of cheap talk \cite{coalitionsafecheap} and in computational games with mediator \cite{costlycom}.

%% file: body.tex
\section{Our Model} 
\label{modelbody}
Our definition of coalition-safe equilibria with virtual payoffs is built on a model of protocol execution that extends the model described in  \cite{backbone}, and is based on \cite{Canetti2,Canetti1,UC2,Canetti4}. This model constitutes the basis for analyzing incentives in an arbitrary blockchain protocol $\Pi$ (but is not necessarily restricted to blockchain protocols). The main components of the model are: a system of interactive Turing machines ITMs ($\mathcal{Z}$,$\mathcal{C}$), a strategic coalition 
of participants that abstractly are referred to as the ``adversary'', $\mathcal{A}$ which is also an ITM, and the ITM instances (ITIs) $P_1,P_2,...,P_n$  that represent the participants of our protocol that run the blockchain protocol $\Pi$. 
 $\mathcal{C}$ is the control program that controls the interactions between the ITIs. $\mathcal{Z}$ is the ``environment'' or in other words the initial Turing machine that represents the external world to the protocol. It gives inputs to the participants and the adversary and it receives outputs from them. The adversary is static and controls a set of $t'$ participants $T\equiv \lbrace  P_{i_1},...,P_{i_{t'}}\rbrace \subseteq \lbrace P_1,...,P_n \rbrace\equiv S$ in the beginning of the execution. 
In the definition of equilibrium we will put forth, we consider executions where 
adversary follows an arbitrary strategy while the remaining participants follow $\Pi$. 

The execution is synchronous and is progressing in rounds as in \cite{backbone}, which means that at the end of each round all the honest participants receive all the messages sent from all the other honest participants. However, compared to \cite{backbone}, instead of just a random oracle on which a cryptographic hash function is modeled, we allow for many oracles where each oracle represents a cryptographic task, such as issuing a digital signature. We denote those by $O_1,\ldots, O_l$. The environment $\mathcal{Z}$ is forced by the control program $\mathcal{C}$ to activate all the participants in  sequence performing a ``round-robin''  participant execution. Each participant can ask each oracle $O_k$ an upper bounded number of queries $q_k$ during each round and each query has a cost $c_k$. The limitation in access is controlled by the control program $\mathcal{C}$. The participants produce messages  delivered via a ``Diffuse Functionality'' as in \cite{backbone}. 
  
The Diffuse functionality adjusts the protocol execution in rounds and determines the communication between the honest participants and the adversary. Specifically it allows the adversary to see the messages produced by the honest participants and delay them until the end of the round. So the adversary can deliver first its messages. However at the end of each round, the Diffuse functionality delivers to all the honest participants all the messages sent from the other honest participants.  Note that the Diffuse functionality gives the opportunity to the adversary to deliver first its own messages  to the honest participants. We provide this capability to the adversary ``for free'', i.e., robustness will be defined even in settings where the  adversary has an inexpensive way of influencing message delivery to its advantage. 


In order to model our notion of  equilibrium we need to compare between two 
possible executions across arbitrary environments. Given this, it is important to fix the number of 
rounds the environment runs the protocol. To accomodate this, we will define 
 as \textit{$r$-admissible} an environment which  performs the protocol a number of rounds $r=p(\kappa)\neq 0$, where $p$ a polynomial, after which it will terminate the execution. Note also that in line with 
\cite{Canetti2,Canetti1} 
the input of the environment will be $1^{p'(\kappa)}$, where $p'$ a polynomial. 
\subsection{The Reward and Cost Functions}

We associate with  a protocol $\Pi$,  a \textit{reward function} that determines the virtual rewards of each set of participants given a local view of a participant that does not belong to the coalition after the last complete round $r$ of the execution. Each participant may have a different local view and as a result different conclusion regarding the rewards of other participants. Note that in a blockchain  protocol this local view is reflected in the blockchain maintained by the participant. Formally:
 $\mathbb{E}$ is the set of all the executions of the protocol $\Pi$ with respect to any adversary and environment. Note that an execution  $\mathcal{E}$  is completely determined by the adversary $\mathcal{A}$, the environment $\mathcal{Z}$, the control program $\mathcal{C}$  and the randomness of these processes, as all the honest participants follow  the protocol $\Pi$. The randomness  determines  the private coins of the participants, the environment, the adversary, and the oracles like the random oracle if they exist as e.g., in \cite{backbone}.
We use $\mathcal{E}_{\mathcal{Z},\mathcal{A}}$ to denote this random variable, where we have specified the environment and the adversary but not the randomness.\footnote{For simplicity we omit reference to the control program because it is the same in all the executions.} 

The function $R^{j}_{T}:\mathbb{E}\longrightarrow \mathbb{R}$ is called the reward function and maps an execution $\mathcal{E}\in \mathbb{E}$ to the virtual rewards of a set $T$ of participants according to the local view of a participant $P_j \in S\setminus T$ after the last complete round $r$ of the execution. 
As an example,  in the Bitcoin blockchain protocol we can consider 
that the rewards for each participant to be  the block rewards from the blocks that it has produced plus the transaction fees of the transactions included in these blocks. We define also
 $R^{\min}_{T}(\mathcal{E}_{\mathcal{Z},\mathcal{A}})\equiv \min\lbrace R^{j}_{T}(\mathcal{E}_{\mathcal{Z},\mathcal{A}})\rbrace _{j:P_j\in S\setminus T}$, and $R^{\max}_{T}(\mathcal{E}_{\mathcal{Z},\mathcal{A}})\equiv \max\lbrace R^{j}_{T}(\mathcal{E}_{\mathcal{Z},\mathcal{A}})\rbrace _{j:P_j \in S\setminus T}$. 
 
The function $ C_i:\mathbb{E}\longrightarrow \mathbb{R}$ is called the cost  function and maps an execution $\mathcal{E}\in \mathbb{E}$  to  the cost of a participant $P_i$ until the end of the last complete  round $r$ of the execution $\mathcal{E}$. Specifically  $C_{i}( \mathcal{E}) =\sum_{k=1}^l c_k\cdot  q_{i,k}(\mathcal{E})$, where $q_{i,k}( \mathcal{E})$ is the number of the queries that $P_i$ asked the oracle $O_k$ until the end of the last complete  round $r$ of  the execution $\mathcal{E}$. Note that $q_{i,k}( \mathcal{E})\leq q_k \cdot r$.\footnote{
Note that the rewards function is defined for a set of participants, but the cost function is defined for a specific participant. In addition we use ``$\equiv$'' to denote  equality of sets, 
random variables and functions.}

\begin{remark}
We assume that rewards and costs are directly comparable and any exchange rate between virtual rewards and cost tokens is constant and is applied directly. Extending our results to a setting where a fluctuating exchange rate in the course of the execution exists between virtual rewards and cost tokens is an interesting direction for future work. 
\end{remark}

\subsection{Utility with Virtual Payoffs}

We next define the (virtual) utility of a  coalition of participants that are controlled by a single rational entity, the adversary. The utility may take various forms and we will consider settings where the adversary cares about   its absolute rewards, its relative rewards or its absolute rewards minus its absolute cost.
Other types of utility  may also be defined, e.g., the adversary may want to minimize the rewards of a specific participant. We will describe the utility of a coalition controlled by a static adversary that includes the set of participants $T\equiv \lbrace  P_{i_1},...,P_{i_{t'}}\rbrace \subseteq \lbrace P_1,...,P_n \rbrace\equiv S$.  

\begin{definition}
We define the utility function of a  $T$-coalition
in the view of the $j$-th participant as a function 
 $U^{j}_T:\mathbb{ E} \longrightarrow \mathbb{R}$ that maps an execution of $\mathbb{ E}$
to a real value. 
\end{definition}   
Based on the above, 
we define also $U^{\max}_T(\mathcal{E}_{\mathcal{Z},\mathcal{A}})  \equiv \max_{j \in S\setminus T}\lbrace U^{j}_T(\mathcal{E}_{\mathcal{Z},\mathcal{A}}) \rbrace$ and $U^{\min}_T (\mathcal{E}_{\mathcal{Z},\mathcal{A}}) \equiv \min_{j \in S\setminus T}\lbrace U^{j}_T(\mathcal{E}_{\mathcal{Z},\mathcal{A}}) \rbrace$. 
Using the reward and cost functions from the previous sections, we define below a few types of utilities that will be relevant in our analysis:

\begin{definition}\label{utilities}
Different types of utility of a coalition $T$ defined over an arbitrary  $\mathcal{E} \in \mathbb{E}$: \begin{itemize} \item {\bf Absolute Rewards.} $U^{j}_T(\mathcal{E}) = R^{j}_{T}(\mathcal{E})$, \item {\bf Absolute Rewards minus Absolute Cost.}
 $U^{j}_T(\mathcal{E})= R^{j}_{T}(\mathcal{E})-\sum_{l:P_l \in T} C_{l}(\mathcal{E}),
$ \item {\bf Relative Rewards.}  $U^{j}_T(\mathcal{E})= \dfrac{ R^{j}_{T}(\mathcal{E})}{ R^{j}_{S}(\mathcal{E})} ,\mbox{ if }   R^{j}_{S}(\mathcal{E})\neq 0 \mbox{ and } 0 \mbox{ otherwise}$.
\ignore{ 
\[ 
U^{j}_T(\mathcal{E})= \left\{
\begin{array}{ll}
      \dfrac{ R^{j}_{T}(\mathcal{E})}{ R^{j}_{S}(\mathcal{E})}, & \mbox{if }   R^{j}_{S}(\mathcal{E})\neq 0 \\
     0, & \text{elsewhere}\\
     
\end{array} 
\right. 
\]
}
\item {\bf Relative Rewards minus Relative Cost.} $U^{j}_T(\mathcal{E})= \dfrac{ R^{j}_{T}(\mathcal{E})}{ R^{j}_{S}(\mathcal{E})}-\dfrac{\sum_{l:P_l\in T}C_{l}(\mathcal{E})}{\sum_{l:P_l\in S} C_{l}(\mathcal{E})}$,\par 
 $\mbox{if }  R^{j}_{S}(\mathcal{E}),\sum_{l:P_l\in S} C_{l}(\mathcal{E})\neq 0
 \mbox{ and } 0 \mbox{ otherwise}$.

\ignore{
\[ 
U^{j}_T(\mathcal{E})= \left\{
\begin{array}{ll}
      \dfrac{ R^{j}_{T}(\mathcal{E})}{ R^{j}_{S}(\mathcal{E})}-\dfrac{\sum_{l:P_l\in T}C_{l}(\mathcal{E})}{\sum_{l:P_l\in S} C_{l}(\mathcal{E})}, & \mbox{if }  R^{j}_{S}(\mathcal{E}),\sum_{l:P_l\in S} C_{l}(\mathcal{E})\neq 0 \\
     0, & \text{elsewhere}\\
     
\end{array} 
\right. 
\]

}
\ignore{
\item $$U^{j}_T(\mathcal{E})= \dfrac{ R^{j}_{T}(\mathcal{E})}{ R^{j}_{S}(\mathcal{E})+\lambda}-\dfrac{\sum_{l:P_l\in T}C_{l}(\mathcal{E})}{\sum_{l:P_l\in S} C_{l}(\mathcal{E})+\lambda}$$
for an arbitrary $\lambda \in (0,1)$ .
}
\end{itemize}

\end{definition}
Note that the total rewards of an execution may be equal to zero. So when we define  relative rewards or relative cost  we should take care that the denominator will never be zero.

\subsection{Coalition Safe Equilibria with Virtual Payoffs}
 \label{subdefinition}

We  will examine two executions of a  protocol with the same environment, but with different adversary and randomness: In the first execution $\mathcal{E}_{\mathcal{Z},H_T}$ the adversary runs the $H_T$ program which controls a set $T$ with cardinality less or equal $t$ and follows the protocol $\Pi$, i.e., plays ``honestly.'' In the second execution 
$\mathcal{E'}_{\mathcal{Z},\mathcal{A}}$ the adversary is denoted by $\mathcal{A}$ and is an arbitrary PPT static adversary that controls the set of users $T$ which includes at most $t$ participants and might deviate in some arbitrary way from the $\Pi$.  For example, in a proof of work blockchain protocol a possible deviation would be to perform  selfish mining~\cite{selfish}.  
 
In more details, $H_T$ is a static adversary that controls a set $T$ of participants and  follows the protocol but it takes advantage of its network presence. 
Note that in our case ``taking advantage of its network presence'' means that the adversary delivers its messages first, when multiple competing solutions/messages (such as proof of work instances) are produced during a round.\footnote{We do not consider in this  present treatment the  cost of having a high presence in the network. Moreover, it is relatively easy to see that if network dominance is given at no cost, it is a rational choice for an adversary  to opt for it in the Bitcoin setting  since it will guarantee that more rewards will be accrued over time. We note that
a similar type of reasoning was adopted also in \cite{zikas} and the corresponding adversary was referred to as ``front running.''}
  \ignore{Even if the adversary plays honestly we allow it to deliver its solutions first, because this in practice is related to how well the network presence of the adversary is and it is expected to use it (we do not consider in this  present treatment the  cost of having a high presence in the network).
\begin{remark}
If we want to make the model more generic  we can take the adversary to have arbitrary network presence. 
\end{remark}
}
 
\begin{definition}
Let $\epsilon,\epsilon' $ be small positive constants near (or equal to) zero and $r$ a polynomial in $\kappa$, the security parameter.
The  protocol is  $(t,\epsilon,\epsilon')$-equilibrium with virtual payoffs (EVP) according to a utility $
\{U^{j}_T\}_{j\in S \setminus T}$ when for every PPT static adversary  $\mathcal{A}$ that controls an arbitrary set $T$ including at most $t$ participants and  for every $r$-admissible environment $\mathcal{Z}$, it holds that 
\[U^{\max}_T(\mathcal{E'}_{\mathcal{Z},\mathcal{A}})\leq  
U^{\min}_T(\mathcal{E}_{\mathcal{Z},H_T})+\epsilon\cdot \mid U^{\min}_T(\mathcal{E}_{\mathcal{Z},H_T}) \mid  + \epsilon' 
\]
with overwhelming probability in $\kappa$.  
$\mathcal{E}_{\mathcal{Z},H_T}$, $\mathcal{E'}_{\mathcal{Z},\mathcal{A}}$ are two independent random variables that represent two independent executions with the same environment $\mathcal{Z}$  and adversary  $H_T$ and $ \mathcal{A}$ respectively. 
\end{definition} 
\begin{remark}
Note that we need absolute value on the right side of the inequality because $ U^{\min}_T(\mathcal{E}_{\mathcal{Z},H_T})$ can be negative when for example it is equal to the profit of a participant. We use two parameters,  $\epsilon$ and $\epsilon'$, to explicitly account for multiplicative and additive deviations in the utility of the diverging adversarial coalition of participants. 
\end{remark} 
\begin{remark}
When the adversary selects the strategy that the  participants controlled by the adversay do not ask any query and do not participate at all, then its utility is zero for all possible choices of utility from  Definition~\ref{utilities}.  As a result  if a protocol is an EVP then this implies that the utility of $H_T$ will be not significantly smaller than $0$. This parallels the participation constraint that is encountered in optimization problems in economics \cite{brit}. 
\end{remark}\ignore{
In the following definition we will take account also that the utility of $H_T$ should be positive or near zero so that the adversary has no incentives not to participate at all.

\begin{definition}
We will say that a protocol is  $(t, \epsilon,\epsilon', \delta')$-Nash-incentive compatible according to a utility function $U^{j}_T$   when:
\begin{itemize}
\item  the protocol is $(t,\epsilon,\epsilon')$-EVP according to  $U^{j}_T$.
\item For every $r$-admissible  environment $\mathcal{Z}$ with input $1^{p'(\kappa)}$ and for every set $T$ with at most $t$ participants it holds \[ U^{\min}_T(\mathcal{E}_{\mathcal{Z},H_T})\geq  - \delta'\] with overwhelming probability in $\kappa$, where $\kappa$ is the security parameter of the system and $ \epsilon,\epsilon',\delta'$ constants such that
 $\epsilon'\in  [ 0, E[\mid U^{\min}_T(\mathcal{E}_{\mathcal{Z},H_T})\mid ]) $, $ \epsilon \in [0,1) $ and  $\delta' \in \mathbb{R}^+$.  
 
\end{itemize}
\end{definition}
Note that if we want to consider that the adversary has no incentives to deviate 
$\epsilon,\epsilon',\delta'$ should be near zero. 
}The definition is generic and includes all probabilistic polynomial time (PPT) static adversaries but in our results we will consider for simplicity a \textit{static PPT adversary with fixed cost} who decides in the beginning how many queries the participants that it controls will ask (and thus how much cost will incur). Recall that this type of cost model in the setting of proof-of-work blockchains is consistent with cloud mining~\cite{cloud}. Formally, we have the following.

\begin{definition}
A static adversary with fixed cost is an adversary that chooses in the beginning of the execution  to control a set  $T\equiv \lbrace  P_{i_1},...,P_{i_{t'}}\rbrace \subseteq \lbrace P_1,...,P_n \rbrace \equiv S$ of $t'$ participants and it commits to the number of queries (of the available $q_k$) each participant $P_{i_m}$, ($m=1,\ldots,t'$) that it controls will ask each oracle $O_k$  during each round of the execution. This number is denoted by $ q_k -x_{m,k}$. This type of adversary can choose any   strategy, but it is committed to paying during each round the cost that it chose in the beginning of the execution.
\end{definition}
\section{Incentives in Bitcoin}
\label{bitcoinbody}

As in \cite{backbone} we will consider that there is only one oracle: the random oracle that models a cryptographic hash function. There are $n$ participants that are activated by the environment in a “round-robin” sequence. When each participant is activated by the environment, it asks at most $q$ queries this oracle. Each query to this random oracle has probability $p$ to give a solution which is a valid block that extends the chain. The messages/solutions are delivered via the Diffuse Functionality. The   expected number of solutions per round by all participants is denoted by $s$. Note that our model is synchronous and  $s$  is  assumed to be  close to zero. 

Regarding Bitcoin with fixed target in a synchronous setting we prove the following results under a PPT static adversary with fixed cost. We will consider that each query to the random oracle has a cost $c$. We suppose that each block gives a \textit{fixed flat reward} $w$ to its creator. Recall $t'$ is the number of the participants controlled by the adversary, $S$ is the set of all the participants and $T$ the set controlled by the adversary.

The results are as follows (the proofs of all the theorems are given in appendix \ref{appendixbitcoin}):

\paragraph{Absolute rewards:} When the utility is based on absolute rewards (cf. Def.\ref{utilities}), then  Bitcoin with fixed target is EVP 
 against a coalition that 
includes even up to all but one of the participants. This is in agreement with the result of \cite{kroll}. The intuition behind this result is that if the adversary cares only about how many blocks it produces then it has no incentives to deviate from the protocol for example by creating forks or by keeping its blocks private. The reason is that if it deviates from the protocol then it increases  the possibility that its blocks will not be included in the public ledger compared to following the protocol. Moreover, the number of the blocks the adversary produces during a round depends only on $p,q,t'$ and not on which chain the adversary extends.  

\begin{theorem}
\label{absolutebitcoin}
For any  $\delta_1 \in(0,0.25)$ such that 
$4 \cdot \delta_1\cdot(1+s)+s<1$, where $s$ the expected number of solutions per round,  Bitcoin with fixed target in a synchronous setting where the reward of each block is a constant, is $(n-1, 4 \cdot \delta_1\cdot(1+s)+s,0)$-EVP according to the utility function absolute rewards (Def.~\ref{utilities}).
\end{theorem}

Note that the better synchronicity we have (the fewer expected number of solutions per round $s$) then the better EVP\footnote{By ``better EVP'' we mean that the actual values of  $\epsilon, \epsilon'$ are smaller. } we have (the lower $4\cdot \delta_1 +(1+4\cdot\delta_1)\cdot s$ is). Recall $4\cdot \delta_1 +(1+4\cdot\delta_1)\cdot s$ is related to how much the adversary can gain if it deviates.

Note that in the theorem we allow the adversary to control all but one of the participants (and not all) because we want at least one honest local chain according to which we can determine the rewards of the adversary.

We extend the above result also in the setting where the block reward changes every at least $l\cdot \kappa$ rounds where $l$ a positive constant and $\kappa$ the security parameter during the execution. 

\begin{theorem}{\label{theorem3}}
Supposing that (i) the block reward changes every at least $l\cdot \kappa$ rounds where $l$ a positive constant and $\kappa$ the security parameter and (ii) the environment terminates the execution at least $l \cdot \kappa$ rounds after the last change of the block reward then it holds: for any  $\delta_1 \in(0,0.25)$ such that 
$4 \cdot \delta_1\cdot(1+s)+s<1$, where $s$ the expected number of solutions per round,   Bitcoin with fixed target in a synchronous setting is $(n-1, 4 \cdot \delta_1\cdot(1+s)+s,0)$-EVP according to the utility function absolute rewards (def.\ref{utilities}).
\end{theorem}

For the proof see Appendix~\ref{rewardchange1}. 

Note that in the analysis above,  we assume throughout  that the target used in the proof of work function remains  fixed as in  \cite{backbone}. It is easy to see that if this does not hold then the adversary using selfish mining \cite{selfish} can cause the protocol to adopt a  target that becomes greater than what is supposed to be and thus the difficulty in mining a block will decrease as the total computational  power would appear smaller than it really is. In this case,  the adversary can produce blocks faster and as such it can magnify its rewards resulting in a negative result in terms of EVP  (see also  \cite{DBLP:journals/corr/abs-1805-08281}). It is an easy corollary that the protocol will not be an EVP in this case. 


\paragraph{Absolute rewards minus absolute cost:}
When the utility is based on absolute rewards minus absolute cost then the Bitcoin protocol with fixed target is EVP  against a coalition that controls even up to all but one of the participants, assuming  the cost of each query $c$ is small enough compared to the block reward $w$. This is in agreement with the result of \cite{zikas}. Again the better synchronicity we have, the better EVP we have.
 
\begin{theorem}
\label{absolutecostbitcoin}
Suppose that there exists $\phi \in (0,1-s)$ such that $c<p\cdot w \cdot \phi/(1+p\cdot q\cdot (n-1))$. Then, supposing that the reward of each block is a constant $w$, it holds: for any $\delta_1 \in(0,0.25)$, such that $c\leq p\cdot w \cdot (1-\delta_1)\cdot \phi/(1+p\cdot q\cdot (n-1))$ and $4 \cdot \delta_1\cdot(1+s)+s<1- \phi$, where $s$ the expected number of solutions per round,  Bitcoin with fixed target in a synchronous setting is $(n-1, (4 \cdot \delta_1\cdot(1+s)+s)/(1- \phi),0)$-EVP according to the utility function absolute rewards minus absolute cost (Def. \ref{utilities}).
\end{theorem}
\begin{remark}

The assumption that there exists $\phi \in (0,1-s)$ such that $c<p\cdot w \cdot \phi/(1+p\cdot q\cdot (n-1))$ means that the reward of each block is high enough to compensate the miners for the cost of the mining. When the cost is high compared to the rewards and the difficulty of mining not fixed then unexpected behaviours appear as proved in \cite{koutsoupias}. 
\end{remark}
Note that the smaller the cost of each query is, the better EVP we have
(because we can select smaller $\phi$ such that $(4\cdot \delta_1 +(1+4\cdot\delta_1)\cdot s)/(1- \phi)$ is smaller). \par We extend the above result also to the case when the block reward changes every at least $l\cdot \kappa$ rounds where $l$ a positive constant and $\kappa$ the security parameter.

\begin{theorem} {\label{theorem6}}
Assume that (i) the block reward changes every at least $l\cdot \kappa$ rounds where $l$ a positive constant and $\kappa$ the security parameter and (ii) the environment terminates the execution at least $l \cdot \kappa$ rounds after the last change of the block reward.
Let $w_j$ for $ j \in \lbrace 0,...,m \rbrace$ be all the block rewards respectively for each player. Assuming that there exists   $\phi \in (0,1-s)$ such that $c<p\cdot w_j \cdot\phi/(1+p\cdot q\cdot (n-1))$ for all $ j \in \lbrace 0,...,m \rbrace$, then it holds: for any $\delta_1 \in(0,0.25)$, such that $c\leq p\cdot w_j \cdot (1-\delta_1)\cdot \phi/(1+p\cdot q\cdot (n-1))$ for all  $ j \in \lbrace 0,...,m \rbrace$ and $4 \cdot \delta_1\cdot(1+s)+s<1 - \phi$, where $s$ the expected number of solutions per round,  Bitcoin with fixed target in a synchronous setting is  $(n-1, (4 \cdot \delta_1\cdot(1+s)+s)/(1- \phi),0)$-EVP according to the utility function absolute rewards minus absolute cost (Def. \ref{utilities}).
\end{theorem}

For the proof see Appendix \ref{rewardchange2}.
\paragraph{Relative rewards:}
When the utility
is based on relative rewards, i.e., the ratio
of rewards of the strategic coalition of the adversary
over the total rewards of all the participants,   
Bitcoin with fixed target cannot be an EVP with small $\epsilon, \epsilon'$. This result is in agreement with \cite{selfish,reasoning}.  
The core idea is to use the selfish mining strategy \cite{selfish,selfish2,selfish3,selfish4,perf}  
to construct an attack that invalidates the equilibrium property. This kind of attack was used also in \cite{backbone} as argument for the tightness of ``chain quality" (chain quality refers to the percentage of the blocks in the public ledger that belong to the adversary). 
Without loss of generality, we will assume that the reward of each block is the same and equal to $w$ (the negative result carries trivially to the general case).\ignore{
 When the utility is equivalent to relative rewards then there is a strategy that deviates  from the protocol and provides
more utility even if the adversary controls one participant. This result is in agreement with \cite{selfish,reasoning}. This selfish mining attack~\cite{selfish,selfish2,selfish3,selfish4,perf} strategy is also described as argument for the tightness of the bound of chain
quality in \cite{backbone}. (Chain quality refers to the percentage of the blocks of the public ledger that belong to the adversary).} The result is in agreement with \cite{selfish} and an argument regarding incentive compatibility of Bitcoin presented in \cite{fruitchain}. However it seems to contradict the result from  \cite{mininggame}, which shows that in a ``strategic-release game'' that describes Bitcoin, honest strategy is Nash equilibrium when the adversary controls a small coalition. This difference arises because 
 that model assumes that all honest miners act as a single miner which implies that when an honest participant produces a block, all the other honest participants adopt this block, something that does not happen in our setting where the adversary is assumed to have network dominance.
 Note that in \cite{backbone,fruitchain} and in our case the adversary has the advantage that it can always deliver its block first and the honest participants adopt the first block they receive. As a result, the blocks of the adversary never become dropped in a case when both the adversary and an honest participant produce a block during a round. 

\begin{theorem}
\label{relativebitcoin}
 Let  $t \in \lbrace 1,...,n-1 \rbrace $ and $t'<\min \lbrace n/2,t+1 \rbrace$. Then for any $ \epsilon +\epsilon' <\dfrac{t'}{n-t'}\cdot(1-\delta') -\dfrac{t'}{n}\cdot (1+\delta'')\cdot (1+s)$, for some $\delta',\delta''$, where $s$ the expected number of solutions per round,  following   Bitcoin with fixed target in a synchronous setting is not  a
$(t,\epsilon, \epsilon')$-EVP  according to the utility function relative rewards (Def. \ref{utilities}).
\end{theorem}

\paragraph{When Transactions Contribute to the Rewards.}
Until now we have supposed that only the flat block reward contributes to the rewards. We next examine what happens when the rewards come also from the transactions included in the mined blocks.

In the description of our model we did not specify the inputs that the environment gives to each participant because these inputs did not contribute to the rewards. We can consider that the inputs are transactions as in \cite{backbone,zikas} and give transactions fees to the participant that will include them in the block that it will produce. The transactions have a sender and a recipient (who can be honest or adversarial participants) and constitute the way in which a participant can pay another participant. So in this setting a participant gains rewards if it produces a block and this block is included in the public ledger (the rewards of each block are the flat reward and the transaction fees) and/or if it is the recipient of a transaction that is included in a block of the public ledger. In this setting the attacks described in \cite{bribery,whale} arise. For example the environment can collaborate with the adversary and send Bitcoin to the participants via the transactions that it gives to them as inputs. Specifically the environment can incentivize the recipients to support an adversarial fork by making these transactions valid only if they are included in this adversarial fork. \par In addition  we can consider that the environment gives the same transactions to all the participants during each round and a transaction cannot be included in more than one block. So if a participant creates an adversarial fork by producing a block that does not include the transactions with high transaction fees then the other participants have incentives to extend it even if they should deviate from the protocol. This happens because in this way they have the opportunity to include the remaining transactions in their blocks and receive the high fees. This attack was described in \cite{selfish3} and  will be more effective when the flat block reward becomes zero and the rewards will come only from the transactions. These observations are in agreement with Theorem $7$ in \cite{zikas} according to which there are some distributions of inputs that make Bitcoin not incentive compatible.
It is an easy corollary to prove that the protocol is not an EVP in this setting. 

\begin{figure}[t]
\begin{center}
\begin{tabular}{ |l|l| }
  \hline
  \multicolumn{2}{|c|}{Notation} \\
  \hline
  $p$ & probability with which a query to the random oracle gives a block \\
  $p_f$& probability with which a query to the random oracle gives a fruit\\
  $q$  & number of queries each participant can ask the random oracle during each round\\
  $t'$ & number of participants controlled by the adversary \\
  $t$  & upper bound of $t'$\\
  $r$  & round after which an execution terminates \\
  $n$  & number of participants \\
  $w$  & flat reward per block (Bitcoin) \\
  $w_f$& flat reward per fruit (Fruitchain \cite{fruitchain})\\
   $s$ & expected number of solutions per round \\
   $x$ & the number of the queries the coalition does not ask during each round \\     
   $S$ & the set of all the participants \\
   $T$ & the set of the participants controlled by the adversary \\ 
  \hline
  
\end{tabular}
\end{center}
\end{figure}
 \section{Incentives in a Fair Blockchain Protocol}
 \label{fruitfair}
In this section we will describe a property, called ``$(t,\delta)$-weak fairness'', which
is sufficient for proving that a protocol is EVP when the utility is based on relative rewards (cf. Def.\ref{utilities}). This property can aid in the design of EVP protocols.  

A protocol will satisfy ``$(t,\delta)$-weak fairness" property when with overwhelming probability  the following hold: firstly when the adversary (which controls at most $t$ participants) deviates, then the fraction of the rewards that the set of all the honest participants gets is at least $(1-\delta)$ multiplied by its relative cost and secondly when the adversary is $H_T$, which means that it follows the protocol, any set of  participants  gets at least  $(1-\delta)$ multiplied by its relative cost.
\begin{definition}
A blockchain protocol satisfies  $(t,\delta)$-weak fairness if for any $r$-admissible environment $\mathcal{Z}$,  for any   PPT adversary $\mathcal{A}$   which controls a set $T$ with at most $t$ participants and for any $j:P_j\in S\setminus T$, where $S$ the set of all the participants, we have with overwhelming probability in the security parameter $\kappa$: \begin{itemize} \item $
R^{j}_{S\setminus T}(\mathcal{E'}_{\mathcal{Z},\mathcal{A}})\geq (1-\delta)\cdot \dfrac{\sum_{l:P_l\in S\setminus T}C_{l}(\mathcal{E}_{\mathcal{Z},H_T})}{\sum_{l:P_l\in S} C_{l}(\mathcal{E}_{\mathcal{Z},H_T})}\cdot R^{j}_{S}(\mathcal{E'}_{\mathcal{Z},\mathcal{A}})$
\item for any subset $S_H \subseteq S $
it holds 
$R^{j}_{S_H}(\mathcal{E}_{\mathcal{Z},H_T})\geq (1-\delta)\cdot \dfrac{\sum_{l:P_l\in S_H}C_{l}(\mathcal{E}_{\mathcal{Z},H_T})}{\sum_{l:P_l\in S} C_{l}(\mathcal{E}_{\mathcal{Z},H_T})}\cdot R^{j}_{S}(\mathcal{E}_{\mathcal{Z},H_T})$
 where $\delta \in [0,1)$.
 \end{itemize}
\end{definition}

Note that $\sum_{l:P_l\in S_H}C_{l}(\mathcal{E}_{\mathcal{Z},H_T})/\sum_{l:P_l\in S} C_{l}(\mathcal{E}_{\mathcal{Z},H_T})$ 
represents the computational power of $S_H$\footnote{$\sum_{l:P_l\in S_H}C_{l}(\mathcal{E}_{\mathcal{Z},H_T})/\sum_{l:P_l\in S} C_{l}(\mathcal{E}_{\mathcal{Z},H_T})= (c \cdot q \cdot r \cdot t_H)/(c \cdot q \cdot r\cdot n)$ where $t_H$ the number of participants of $S_H$}, because honest participants and $H_T$ ask all the queries during each round. In addition $\sum_{l:P_l\in S} C_{l}(\mathcal{E}_{\mathcal{Z},H_T})\neq0$ as the execution lasts at least one round. We do not divide with $ R^{j}_{S}(\mathcal{E}_{\mathcal{Z},H_T}) $ as we do not exclude the case that is equal to zero.\par

According to the following theorem when a protocol satisfies the $(t,\delta)$-weak fairness property and the total rewards are greater than zero with overwhelming probability then following the protocol is EVP under an adversary that controls at most $t$ participants. This theorem will be also used in order to prove that the Fruitchain protocol \cite{fruitchain} is EVP when the utility is based on relative rewards.

\begin{theorem}
\label{theorem8}

When a   protocol satisfies $(t,\delta)$-weak fairness  and in addition for any $j: P_j\in S\setminus T $, for any PPT   adversary $\mathcal{A}$ which controls a set $T$ with at most $t$ participants and for any $r$-admissible environment $\mathcal{Z}$ it holds $R^{j}_{S}(\mathcal{E'}_{\mathcal{Z},\mathcal{A}})> 0$ with overwhelming probability in the security parameter $\kappa$, then following the protocol is $( t,0,\delta)$-EVP according to the utility function relative rewards (def.\ref{utilities}). 
\end{theorem}  For the proof see Appendix \ref{fruitproof1}.

\paragraph{Comparison between $(t,\delta)$-weak fairness and other notions:} Our property is weaker than \textit{$(T,\delta)$-approximate fairness w.r.t. $\rho$ attackers} defined in \cite{fruitchain} and \textit{ideal chain quality} defined in \cite{backbone}.\par
The property \textit{$(T,\delta)$-approximate fairness w.r.t. $\rho$ attackers} defined in \cite{fruitchain} says that in any sufficient long window of the chain with $T$ blocks, any set of honest participants with computational power $\phi$ will get with overwhelming probability at least $(1-\delta)\cdot \phi$  fraction of the blocks regardless what the adversary with a fraction of computational power at most $\rho$ does.\par 
\textit{Ideal chain quality} defined  in \cite{backbone} says that  any coalition of participants (regardless the mining strategy they follow) will get a percentage of blocks in the blockchain that is   proportional to their collective hashing power.
\par

Our property is weaker than \textit{$(T_0,\delta)$-approximate fairness  w.r.t. $t/n$ attackers} ($n$ is the number of all the participants)\footnote{ To be precise it is weaker than fairness under the restriction that the environment performs the protocol so many rounds that with overwhelming probability any honest participant has a local chain of length at least $T_0$. \ignore{ $T_0=T(\kappa)/ \delta$} Note that this happens because in our definition we have not used $T_0$ as parameter.}  defined in \cite{fruitchain} and \textit{ideal chain quality} in \cite{backbone} in the sense that when the adversary deviates from the protocol we demand that only the whole set of the honest participants gets a fraction of rewards at least $(1-\delta)$ multiplied by its relative cost, not all the subsets of the honest participants. In the same way our definition is also weaker than \textit{race-free property} defined in \cite{tor}\footnote{Note that when a cryptocurrency is pseudonymous and not anonymous then it is difficult to secure that every subset of honest participants will take the appropriate percentage of the blocks, because maybe it is the case where the adversary cannot decrease much the percentage of the blocks that belongs to the whole set of the honest participants, but it can act against a specific participant with some characteristics  revealed from the graph of the transactions. For example there are some works that analyze the statistical properties of the Bitcoin transaction graph  and describe identification attacks in Bitcoin, \cite{anonymity1,anonymity2}}. 

\section{Incentives in the Fruitchain Protocol}
\label{fruitabstract}
In this section, we analyze incentives of \cite{fruitchain}. As before we assume the participants use a hash function which is modeled as a random oracle. The number of the queries to the random oracle by each participant during a round  is bounded by $q$. Let the total number of the participants be $n$. Each query to the random oracle can give with probability $p$ a block and with probability $p_f$ a fruit, where $p_f$ is assumed to be  greater than $p$. This is achieved via the 2-for-1 POW technique of \cite{backbone}.  At the beginning of each round, when the honest participants are activated, they ``receive" the fruits and the blocks from the Diffuse Functionality, they choose the chain that they will try to extend and they include in the block they try to produce ``a fingerprint" of all the ``recent" fruits (as defined in \cite{fruitchain}) that have not been included in the blockchain yet. Then they ask the random oracle $q$ queries. When an honest participant finds a fruit or a block, it gives it to the Diffuse Functionality and it continues asking the remaining queries. Even if it finds more than one fruit during a round, it gives all the fruits to the Diffuse Functionality. The adversary is activated at the end and it can ask $t' \cdot q$ queries, where $t'$ is the number of the participants that it controls.  
We consider that the rewards come only from the fruit\footnote{Note that in the Fruitchain protocol \cite{fruitchain} the fainess property holds for the fruits; actual blocks are possibly still vulnerable to selfish mining attacks \cite{selfish}. So if we consider that also the blocks give a flat reward then we cannot use the fairness property proved in \cite{fruitchain}. } and the difficulty in mining a block is fixed. In our case each query to the random oracle has a cost $c$. In the proofs  we will assume that the adversary is static, the model is synchronous and the Diffuse Functionality works as \cite{backbone}, and each fruit gives reward equal to $w_f$. 
\paragraph{Relative rewards:} According to the following theorem if the adversary controls fewer than half of the participants and wants to maximize its \ignore{joint}relative rewards which means that its utility is based on relative rewards (Def.~\ref{utilities}), then following the Fruitchain protocol is EVP. \ignore{In addition as the protocol lasts at least one round, relative rewards are positive with overwhelming probability and as a result Fruitchain protocol is also Nash- incentive compatible.}This theorem allows us to understand in a formal way how mining simultaneously fruits and blocks can eliminate the impact of selfish mining \cite{selfish} on the incentive compatibility of the protocol. We note that the core advantage stems from the 2-for-1 POW  technique used for simultaneous mining which was initially proposed for the mitigation of selfish mining in \cite{backbone} in the context of achieving Byzantine agreement for honest majority and later was adapted in \cite{fruitchain} for a similar purpose in the context of fair blockchains. \ignore{Note that in \cite{fruitchain} the utility that is used is absolute and not relative rewards.} 

\begin{theorem}{\label{relativefruitchain}}
Let $\delta \in (0,1)$ and $T_0$ such that the Fruitchain protocol satisfies $(T_0, \delta)$-approximate fairness property. Then the Fruitchain protocol is $(n/2-1,0,\delta)$-EVP according to the utility function relative rewards (Def. \ref{utilities}), under an $r$-admissible environment where $r \geq T_0/ (p_f \cdot (\frac{n}{2}+1) \cdot (1- \delta)\cdot q)$. 
\end{theorem}
For the proof see Appendix \ref{fruitproof2}. It uses Chernoff bound and Theorem \ref{theorem8}.

Note that for any $\delta \in (0,1)$ and appropriate $T_0$ 
the Fruitchain protocol satisfies $(T_0,\delta)$-approximate fairness property (Subsection 4.2 in \cite{fruitchain})\footnote{In \cite{fruitchain} the number of queries $q$ each participant can ask during each round is $1$.}. 
\begin{remark}
The above theorem holds also when we take into account also the transaction fees from each fruit and at the end of the execution we distribute evenly the total rewards among the miners of the fruits (as assumed in \cite{fruitchain}\footnote{ To be precise in \cite{fruitchain} the rewards of each fruit are shared evenly among the miners of the fruits included in a long enough preceding part of the chain.}).
\end{remark}

\paragraph{Absolute rewards minus absolute cost:}  We will prove that the Fruitchain \cite{fruitchain} protocol in a synchronous setting is EVP according to utility based on absolute rewards minus absolute cost (Def. \ref{utilities}) if the adversary controls all but one participants, when the cost of each query $c$ is small enough compared to the reward of each fruit $w_f$.  Note that the smaller the cost of each query is, the better  EVP we have.\footnote{Note that here synchronicity does not affect how good the EVP is in contrast to our theorems regarding  Bitcoin with fixed target. This happens because when honest participants find more than one fruit during a round, all of them can be included in the chain eventually, in contrast to  Bitcoin where when two honest participants find a block during a round then only one of them can be included in the chain.} 
  \par The intuition behind the proof is that (i) the rewards come from the fruits that are produced by mining and (ii) the total number of the fruits the adversary can produce is bounded (with overwhelming probability) whatever strategy it follows. So if the adversary can have this number of fruits even if it follows the protocol, it has no reason to deviate.

\begin{theorem}
{\label{absolutefruitchain}}
Assume that each fruit gives a constant reward and there exists  $\phi \in (0,1)$ such that $c<p_f\cdot w_f \cdot \phi$. Then  for any $\delta_1 \in (0,0.25)$, such that $c\leq p_f\cdot w_f \cdot (1-\delta_1)\cdot \phi$  and  $4 \cdot \delta_1<1- \phi$
the Fruitchain protocol in a synchronous setting is  $(n-1, 4 \cdot \delta_1/(1- \phi), 0)$-EVP according to the utility function absolute rewards minus absolute cost (Def. \ref{utilities}) .
\end{theorem}
For the proof see Appendix \ref{fruitproof3}.
\begin{remark}
The assumption that there exists $\phi \in (0,1)$ such that  $c<p_f\cdot w_f \cdot \phi$ means that the reward of each block is high enough to compensate the miners for the cost of the mining. Finally note that trivially if we consider that $c=0$ then the assumption of the above theorem holds for $\phi$ close to zero and the utility is just absolute rewards (Def. \ref{utilities}).
\end{remark}

%% file: appendix.tex
\section{Game Theoretic Notions}
\label{definitions}
\begin{itemize}
\item A strategy profile, which indicates how each participant behaves in the game,  is an \textit{$\epsilon$-Nash equilibrium} when the following holds: if all but one of the participants follow their strategy indicated by the strategy profile,
the remaining participant has no incentives to deviate 
from its indicated strategy as well, as its utility can only  be increased  by a small insignificant amount bounded by $\epsilon$, see e.g., \cite{essential}. Extended notions of equilibria capture strategic coalitions as well, cf.  \cite{coalition,strongn}, giving rise to ``Strong'' Nash Equilibria. Note that if we show that a blockchain protocol is an $\epsilon$-Nash equilibrium, we know that nobody  has the incentive to deviate from the protocol, if everybody else follows the protocol.
\item  The concept of \textit{Incentive compatibility} appears in a few different forms in the literature. 
 ``Dominant-strategy incentive-compatibility" is satisfied when there is not a strictly better strategy than telling the truth or following the protocol respectively whatever the other participants do. ``Bayesian-Nash incentive-compatibility"  is a  weaker notion  and a protocol satisfies it when there is  a type of Nash equilibrium called ``Bayesian Nash equilibrium", where all the participants tell the truth supposing that all the other participants do the same \cite{multi}. In cryptocurrency literature some times the incentive compatibility notion is used as equivalent to the Nash equilibrium notion \cite{leonardos2019presto}. More broadly, maximizing the profits or maximizing the utility can be seen as an \textit{optimization problem} that includes at least two constraints. The first constraint is \textit{incentive compatibility} and the second constraint is the \textit{participation constraint} which suggests  that when a participant participates in the game, this does not result in  lower utility compared to not participating \cite{brit}.
 \end{itemize}
\section{Chernoff Bounds}
\label{chernoff}
Let $X_i:i\in \lbrace 1,..., n\rbrace$ are mutually independent Boolean random variables and $\forall i Pr(X_i=1)=p$. Let $X=\sum_{i=1}^n X_i$ and $\mu=pn$.
Then we have for any $\delta \in (0,1]$
$$ Pr(X\leq(1-\delta)\mu)\leq e^{-\delta^2\mu/2}$$ and $$ Pr(X\geq(1+\delta)\mu)\leq e^{-\delta^2\mu/3}$$

\section{The Theorems and Proofs Regarding Incentives in Bitcoin}
\label{appendixbitcoin}
In this section we will use our definition to examine if  the Bitcoin with fixed target in a synchronous setting is EVP according to different utilities under a PPT static adversary with fixed cost, when each query to the random oracle has a cost and the difficulty in mining blocks (or in other words the target of each block) is fixed. The block reward will be fixed or will change every at least $l\cdot \kappa$ rounds, where $l$ a positive constant and $\kappa$ the security parameter, as we do not take into account transaction fees (we consider only the flat reward).  \par 
In more detail, as in \cite{backbone}\footnote{We will use some notation and some proof techniques from \cite{backbone}.} we will consider that there is only one oracle: the random oracle. The difficulty in mining each block is fixed. Each honest participant asks during each round $q$ queries the random oracle. In our case each query to the random oracle has cost $c$ and not zero. The probability with which a query is successful is $p$ and $n$ is the number of the participants. Let $s$  be the expected number of solutions per round and as the model is synchronous, it is near  zero. The security parameter is $\kappa$ which is the domain of the hash function. \par The adversary is static and it controls an arbitrary set  $T\equiv \lbrace  P_{i_1},...,P_{i_{t'}}\rbrace$ with $t'$ participants. Let $x_{m}$ be the queries the participants controlled by the adversary $P_{i_m}$  will not ask the random oracle during each round. $x=\sum_{m=1}^{t'} x_m$  is the total number of the queries that all the participants controlled by the adversary collectively do not ask during each round. Note that $x$ is a constant not a random variable as it is determined in the beginning by the static adversary with fixed cost. It holds $0\leq x \leq q\cdot t' $.  \par  Let $R^{j}_{T}(\mathcal{E}$) be the rewards of the blocks that are produced by $T$ and are included in the local chain of $P_j$ after the last complete round $r$ of the execution $\mathcal{E}$. \par
Some clarifications regarding \cite{backbone} that are useful for our proofs are described in the following subsection. 
\subsection{Clarifications about Proofs in Bitcoin}

The honest participants ask during each round all the available $q$ queries even if they find a block in the middle of the round. If the honest participants find more than one block during a round they give to the Diffuse functionality only the first block. In addition even if an honest participant receives from the Diffuse Functionality in the middle of a round a block produced by another participant, it will not change the block that it tries to extend. As a result, although  forks with two or more honest blocks are permitted, a chain cannot be extended by two honest blocks in a round. On the other hand, these restrictions do not hold for the  participants controlled by the adversary because if a  participant controlled by the adversary finds more than block during a round it can give all the blocks to the Diffuse functionality. In addition, the honest participants choose the first block they receive in the case of a tie, which means that the adversarial blocks are always preferred by the honest participants, as the adversary can deliver its block first.

\par 
 Successful round (defined in \cite{backbone}) for a subset of participants  is a round where at least one of the participants included in this subset has found  a solution.
The following lemma is an extension of ``Chain-Growth" Lemma $7$ of \cite{backbone}. 
\begin{lemma}{\label{firstlemma}}
For every $r$-admissible environment $\mathcal{Z}$ with input $1^{p'(\kappa)}$, at the end of each round of an execution $\mathcal{E}_{\mathcal{Z},H_T}$, the local chains of all the honest participants have the same number of blocks that is equal to the successful rounds for all the participants.
\end{lemma}



\begin{proof}
This can be proved by induction on the round $r$ using the clarifications above  regarding the honest participants and using the fact that every participant controlled by $H_T$ follows  the protocol. \par 
Let an arbitrary execution  $\mathcal{E}_{\mathcal{Z},H_T}$.
For the basis $r=1$ : if the first round is not a successful round then all the participants have a local chain with length zero equal to the number of the successful rounds which is also zero. If the first round is successful then all the participants have a local  chain of length $1$. This holds because:

\begin{itemize}
\item the participants  cannot have at the end of the first round  a local  chain with length zero as all the participants at the end of the first round will receive from the Diffuse Functionality all the blocks produced during the first round. Note that $H_T$ follows the  protocol and always sends its  blocks to the Diffuse functionality.
\item the participants cannot have at the end of the first round  a local chain with more than one block given that even if more than one block have been produced during the first round, these blocks can extend the length of the local chains only by one. This holds because (i) if the participants (also the participants that are controlled by $H_T$) find more than one block they give to the Diffuse functionality only the first block and (ii) even if a participant receives from the Diffuse Functionality in the middle of a round a block produced by another participant, it will not change the block that it tries to extend.
\end{itemize}
For the induction step we suppose that at the end of the round $r$ all the participants have local chains with length equal to the successful rounds and we can prove with the same arguments that at the end of round $r+1$, if the round $r+1$ is successful, all the participants will extend their chain by one block.
\end{proof}
Note that at the end of each round  although the local chains of participants will have the same length, they may contain different blocks in the last part, because forks with two or more honest blocks are permitted.

\begin{lemma}{\label{secondlemma}}
For every $r$-admissible environment $\mathcal{Z}$ with input $1^{p'(\kappa)}$, at the end of each round of an execution $\mathcal{E}_{\mathcal{Z},H_T}$, the number of the blocks that are produced by the set $T$ of the participants controlled by the adversary and are included in a local chain of an arbitrary honest participant are equal to the number of the successful rounds for $T$.
\end{lemma}

\begin{proof}
This can be proved also by induction on the round $r$ taking
into account the fact that $H_T$ delivers always its blocks first and the participants adopt the first block they receive. Specifically when both a participant controlled by the adversary and an honest participant have produced a block during a round then all the participants will adopt the block produced by the participant controlled by the adversary.\par Note that $H_T$ will give to the  Diffuse Functionality at most one block per round for each participant controlled by the adversary and even if there are more than one block produced by $H_T$ during a round they extend the length of the chains by one. So the local chains of the honest participants at the end of each round may contain different blocks produced by the adversary, but all will have the same number of blocks produced by the adversary. \par
Let an arbitrary execution  $\mathcal{E}_{\mathcal{Z},H_T}$. 
This can be proved by induction on the round $r$.  For the basis $r=1$ : if the first round is not a successful round for  $T$ then all the participants have a local chain with zero blocks produced by the participants controlled by the adversary which is equal to the number of the successful rounds for $T$ that is also zero. If the first round is successful for $T$ then all the participants  have a local  chain that includes exactly one  block produced by the adversary $H_T$. This holds because:

\begin{itemize}
\item the participants cannot have at the end of the first round  a local chain with no block produced by $T$, because: 

\begin{itemize}
\item $H_T$ follows the protocol and always sends its  blocks to the Diffuse functionality which means that all participants  receive its blocks. 
\item the participants  will adopt a block produced by $T$ at the end of the first round  even if another participant has also produced a block during the first round because $H_T$ delivers its blocks first.
\end{itemize}

\item the participants  cannot have at the end of the first round  a local chain with more than one block produced by $T$ because even if more than one block have been produced by $H_T $ during the first round these blocks can extend the length of the  local chains only by one given that $H_T$ follows  the protocol.
\end {itemize}
For the induction step we suppose that at the end of the round $r$ all the honest participants  have local chains that include blocks produced by $T$ whose number is equal to the successful rounds for $T$. Then:  
\begin{itemize}
    \item  If round $r+1$ is not a successful round for $H_T$ then the number of the blocks produced by $T$ that are included in the local chain of an arbitrary honest participant does not change. 
    \item If round $r+1$ is successful for $T$ then all the honest participants  include exactly one more block produced by $T$, not necessary the same, because of the arguments described above.
\end{itemize}
\end{proof}
\subsection{Absolute Rewards}
In this subsection we examine if Bitcoin is EVP when utility is equivalent to absolute rewards, which means  $U^{j}_T\equiv R^{j}_{T}$.
\ignore{In more detail, we prove that if a static adversary with fixed cost wants to maximize the flat reward of the blocks that it has produced and are included in the public ledger, then it has no incentives to deviate  from the Bitcoin with fixed target ina synchronous setting, even if it controls $n-1$ participants . This happens because if it does, it will earn with overwhelming probability in the security parameter an amount very close to zero more compared to following the protocol.} \par
Our theorem assumes that the block reward is fixed and equal to $w$. However it holds also when we assume that (i) the block reward changes every at least $l\cdot \kappa$ rounds where $l$ a positive constant and $\kappa$ the security parameter and (ii) the environment terminates the execution at least $l \cdot \kappa$ rounds after the last change of the block reward. The exact theorems and proofs of this case are given in  the next subsection. 
The intuition is that the number of the successful rounds of a period is independent of the number of the successful rounds of a following period with different block reward. The same it holds for the number of blocks produced by a set of participants.

By Lemma $\ref{secondlemma}$ we can conclude that 
\begin{lemma}{\label{min=max}}
  For every $r$-admissible environment $\mathcal{Z}$ with input $1^{p'(\kappa)}$, where $\kappa$ the security parameter it holds  
$$ R^{\max}_{T}(\mathcal{E}_{\mathcal{Z},H_T})\equiv  R^{\min}_{T}(\mathcal{E}_{\mathcal{Z},H_T})\equiv X_r^T (\mathcal{E}_{\mathcal{Z},H_T})\cdot w$$ where  $X_r^T (\mathcal{E}_{\mathcal{Z},H_T})$ are the successful rounds for $T$ until the last complete round $r$ of execution $\mathcal{E}_{\mathcal{Z},H_T}$.
\end{lemma}
\begin{proof}
The rewards of $T$ according to the local chain  of an honest participant $P_j$ come from the flat reward of each block included in this local chain that is produced by a participant of $T$. Moreover the flat reward of all the blocks gives the same amount of Bitcoin equal to $w$. \par By Lemma $\ref{secondlemma}$, at the end of the last complete round $r$ of execution $\mathcal{E}_{\mathcal{Z},H_T}$ all the honest participants  have local chains whose number of blocks produced by $T$ is equal to the successful rounds for $H_T$ at the end of the round $r$.  So the maximum reward of $T$ is equal to the minimum reward, as all the local chains of all the honest participants  contain the same number of blocks produced by $T$, and it is equal to the successful rounds for $T$ at the end of round $r$ multiplied by $w$. 

\end{proof}

\begin{theorem*} 
For any  $\delta_1 \in(0,0.25)$ such that 
$4 \cdot \delta_1\cdot(1+s)+s<1$, where $s$ the expected number of solutions per round, the Bitcoin with fixed target in a synchronous setting where the reward of each block is a constant, is $(n-1, 4 \cdot \delta_1\cdot(1+s)+s,0)$-EVP according to the utility function absolute rewards (def.\ref{utilities}) .
\end{theorem*}
Note that our model is synchronous and as a result the expected number of solutions per round $s$  are close to zero. 
\begin{proof}
Let arbitrary $\delta_1 \in(0,0.25)$ such that $4\cdot \delta_1 +(1+4\cdot\delta_1)\cdot s<1 $.
We choose also an arbitrary $r$-admissible environment $\mathcal{Z}$ with input $1^{p'(\kappa)}$, where $\kappa$ the security parameter and an arbitrary adversary $\mathcal{A}$ static with fixed cost  that is PPT and it controls an arbitrary set $T$ with $t'$ participants  where $t' \in \lbrace 1,...,n-1\rbrace$.
Note that when the adversary controls $0$ participants  then the theorem is proved trivially as the utility of the adversary is zero regardless its strategy.\par
We will examine two executions of the Bitcoin  with the same environment, but with different adversary : In the first execution $\mathcal{E}_{\mathcal{Z},H_T}$ the adversary is $H_T$ and in the second execution 
$\mathcal{E'}_{\mathcal{Z},\mathcal{A}}$ the adversary is $\mathcal{A}$.
Note that the environment is the same in the two executions, which means that it gives the same inputs to the participants and it sends the same messages to the adversary, although it will receive different responses from the adversary and  specifically it will receive no response from $H_T$. In addition the environment will decide before the start of the execution the round $r=p(\kappa)\neq 0$ after which it will terminate the protocol. So the two executions will last the same number of rounds as they have the same environment.\par
In more detail, $H_T$ follows  protocol and ignores the messages that receives from the environment $\mathcal{Z}$. 
So in the execution $\mathcal{E}_{\mathcal{Z},H_T}$ the environment can do only the following:
\begin{itemize}
\item It gives transactions as input to all the participants.
\item It can send messages to $H_T$, but $H_T$ will ignore it.
\item It has decided when $\mathcal{E}_{\mathcal{Z},H_T}$  ended.
\item It receives outputs from the participants.
\end{itemize}
Firstly by Lemma $\ref{min=max}$ and by Chernoff bound we have that:
\begin{equation}{\label{H_T}}
  U^{\min}_T(\mathcal{E}_{\mathcal{Z},H_T})\equiv   R^{\min}_{T}(\mathcal{E}_{\mathcal{Z},H_T})\equiv X_r^T (\mathcal{E}_{\mathcal{Z},H_T})\cdot w > \dfrac{p\cdot q \cdot t'}{1+p \cdot q \cdot t'}\cdot r\cdot (1-\delta_1)\cdot w> 0
\end{equation}
with overwhelming probability in $r$ and as $r=p(\kappa)$ also in $\kappa$.
In more detail this can be proved as follows :
\begin{itemize}
\item $X^{T,m}(\mathcal{E}_{\mathcal{Z},H_T})$ is a Boolean random variable, where $X^{T,m}(\mathcal{E}_{\mathcal{Z},H_T})=1$ when round $m$ was successful for  $H_T$. The variables $\lbrace X^{T,m}(\mathcal{E}_{\mathcal{Z},H_T}) \rbrace_{m=1,...,r} $ are independent Bernoulli trials. 
\item $X^T_r(\mathcal{E}_{\mathcal{Z},H_T})\equiv \sum_{m=1}^{r}X^{T,m}(\mathcal{E}_{\mathcal{Z},H_T})$ is the number of the successful rounds for $H_T$ until the last complete round $r$ of $\mathcal{E}_{\mathcal{Z},H_T}$.
\item $\forall m $ $E[X^{T,m}(\mathcal{E}_{\mathcal{Z},H_T})]= 1-(1-p)^{qt'}$, where $p$ is the probability with which one query to the random oracle is successful and $q$ is the number of the queries that each participant can ask the oracle during each round. We consider $E[X^{T,m}\mathcal{E}_{\mathcal{Z},H_T})]$ as constant. Note that $H_T$ asks all the available queries.
\end{itemize}
By Lemma $\ref{min=max}$ we have that: 
\begin{align}{\label{T,H_T}}
R^{\min}_{T}(\mathcal{E}_{\mathcal{Z},H_T})\equiv R^{\max}_{T}(\mathcal{E}_{\mathcal{Z},H_T})\equiv X^T_r(\mathcal{E}_{\mathcal{Z},H_T})\cdot w
\end{align}

By Chernoff bound we have that for any $\delta_2\in(0,1)$ and as a result also for $\delta_1$:
\begin{equation} {\label{ber}}
Pr[X^T_r(\mathcal{E}_{\mathcal{Z},H_T}) > (1-\delta_1)\cdot( 1-(1-p)^{qt'})\cdot r]\geq 1-e^{-\dfrac{(\delta_1)^2 \cdot( 1-(1-p)^{qt'})\cdot r}{2}}
\end{equation} 

\par In addition with probability $1$ we will prove the following that is stated in \cite{backbone}:

\begin{equation}{\label{bound}}
(1-\delta_1)\cdot( 1-(1-p)^{qt'})\cdot r\geq \dfrac{p\cdot q \cdot t'}{1+p \cdot q \cdot t'}\cdot r\cdot (1-\delta_1)
\end{equation}
Specifically, it holds $1-p\leq e^{-p}$ and as a result
$1-(1-p)^{q\cdot t'}\geq 1-e^{-p\cdot q\cdot t'}$.\par  
Moreover, we have that $1-e^{-x}\geq x/(1+x)$ for $x\geq 0$ (here $x=p\cdot q\cdot t'$).
 and as a result :
$$ 1-(1-p)^{qt'}\geq 1-e^{-p\cdot q\cdot t'}\geq \dfrac{p\cdot q \cdot t'}{1+p \cdot q \cdot t'}$$
The above inequalities can be proved taking the functions $f(x)=(1-e^{-x})\cdot (1+x)-x$ and $g(x)= e^{-x}-1+x$ for $x\geq 0$ and studying their minimum value using their monotony. \par So by the above inequality ($\ref{bound}$), by Lemma $\ref{min=max}$ and by equation $(\ref{ber})$ we conclude equation ($\ref{H_T}$) . \par
 In addition  for any $\delta \in(0,1)$ and as a result also for $\delta_1$ it holds with overwhelming probability in $r$ and also in $\kappa$ that:
 \begin{equation}{\label{bad}}
   U^{\max}_T(\mathcal{E'}_{\mathcal{Z},\mathcal{A}}) \leq Z_r(\mathcal{E'}_{\mathcal{Z},\mathcal{A}})\cdot w <p\cdot q\cdot t'\cdot r\cdot (1+\delta_1)\cdot w
  \end{equation} where $Z_r(\mathcal{E'}_{\mathcal{Z},\mathcal{A}})$  is the number of the blocks the adversary has produced until the last complete round $r$ of $\mathcal{E'}_{\mathcal{Z},\mathcal{A}}$.

This can be proved with Chernoff bound taking into account the fact that the adversary cannot gain rewards from more blocks than these it has produced.
In more detail,  
\begin{itemize}
\item $Z_{i,j,k}(\mathcal{E'}_{\mathcal{Z},\mathcal{A}}) $ is a boolean random variable  and $Z_{i,j,k}(\mathcal{E'}_{\mathcal{Z},\mathcal{A}}) =1$ when at round $i$ of $\mathcal{E'}_{\mathcal{Z},\mathcal{A}}$ the $j$-th query to the random oracle  of the $k$-th  participant controlled by the adversary  is successful. $Z_{i,j,k}(\mathcal{E'}_{\mathcal{Z},\mathcal{A}})$ are independent Bernoulli trials.
\item $Z_r(\mathcal{E'}_{\mathcal{Z},\mathcal{A}})\equiv \sum_{i=1}^{r}\sum_{k=1}^{t'}\sum_{j=1}^{q-x_k} Z_{i,j,k}(\mathcal{E'}_{\mathcal{Z},\mathcal{A}})$. $$R^{\max}_T(\mathcal{E'}_{\mathcal{Z},\mathcal{A}}) \leq Z_r(\mathcal{E'}_{\mathcal{Z},\mathcal{A}})\cdot w   $$ as the adversary cannot  gain rewards for more blocks than these it  has produced.
\item $E[Z_{i,j,k}]=p$
 \end{itemize}

 By Chernoff bound we have that for any $\delta \in(0,1)$ thus also for $\delta_1$
\begin{equation}{\label{adrew}}
 Pr[Z_r(\mathcal{E'}_{\mathcal{Z},\mathcal{A}})< p\cdot( q\cdot t'-x)\cdot r\cdot (1+\delta_1)]\geq 1-e^{-\dfrac{(\delta_1)^2\cdot p\cdot (q\cdot t'-x)\cdot r}{3}}
\end{equation} 
In addition with probability $1$ it holds:
$$p\cdot( q\cdot t'-x)\cdot r\cdot (1+\delta_1)\leq p\cdot q\cdot t'\cdot r\cdot (1+\delta_1)$$

By the above equation we can conclude $(\ref{bad})$.  \par
Finally by equations (\ref{H_T}),(\ref{bad}) we have that 

\begin{equation} {\label{final}}
U^{\max}_T(\mathcal{E'}_{\mathcal{Z},\mathcal{A}}) \leq U^{\min}_T(\mathcal{E}_{\mathcal{Z},H_T})\cdot (1+4\cdot \delta_1 +(1+4\cdot\delta_1)\cdot s)
 \end{equation}
with overwhelming probability in $r$ and also in $\kappa$.\par
 
 \par In more detail this can be proved as follows: 
\begin{itemize}
\item Let $F$ be the final event where it holds $$ U^{\max}_T(\mathcal{E'}_{\mathcal{Z},\mathcal{A}}) \leq U^{\min}_T(\mathcal{E}_{\mathcal{Z},H_T})\cdot  (1+4\cdot \delta_1 +(1+4\cdot\delta_1)\cdot s)$$ We want to prove that $Pr[F]\geq 1-negl(r)$.
\item Let $A$ be the event where $$U^{\min}_T(\mathcal{E}_{\mathcal{Z},H_T})> \dfrac{p\cdot q \cdot t'}{1+p \cdot q \cdot t'}\cdot r\cdot (1-\delta_1)\cdot w $$ By ($\ref{H_T}$) we have $Pr[A]\geq 1-negl(r)$.
\item Let $B$ be the event where $$ U^{\max}_T(\mathcal{E'}_{\mathcal{Z},\mathcal{A}})< p\cdot q\cdot t'\cdot r\cdot (1+\delta_1)\cdot w$$ By ($\ref{bad}$) we have that $Pr[B]\geq 1-negl(r)$.
\item Using the above statements we have that $$Pr[A\cap B]=Pr[A]-Pr[A\cap\neg B] \geq 1-negl(r)$$
At this point in order to prove ($\ref{final}$) we only have to prove that $$Pr[A\cap B] \leq Pr[F]$$
In order to prove the above statement we will suppose that the event
$A\cap B$ holds and we will prove that the event $F$ holds.\par So we have that when $A\cap B$ holds  then:
 \end{itemize}

\begin{align*}
 U^{\max}_T(\mathcal{E'}_{\mathcal{Z},\mathcal{A}})  &\leq p\cdot q\cdot t'\cdot r\cdot (1+\delta_1)\cdot w \\ &\leq \dfrac{p\cdot q \cdot t'}{1+p \cdot q \cdot t'}\cdot r\cdot (1-\delta_1)\cdot w  \cdot (1+4\cdot \delta_1 +(1+4\cdot\delta_1)\cdot s) \\& <  U^{\min}_T(\mathcal{E}_{\mathcal{Z},H_T}) \cdot (1+4\cdot \delta_1 +(1+4\cdot\delta_1)\cdot s)
\end{align*} 
\par This means that when
$A\cap B$ holds then also $F$ holds that is what we want to prove. \par  Note that 
$$\dfrac{1+\delta_1}{1-\delta_1}\leq 1+4\cdot \delta_1$$ as 
$\delta_1 \in (0,0.25)$. This can be proved if we find the minimum of the
function $$f(x)=-(1+x)/(1-x)+1+4x $$ by its monotony.
\par Moreover $p\cdot q\cdot t'\leq p\cdot q\cdot n=s$, where $s$ the expected number of solutions  of all the participants per round.  Note that when the system is synchronized $s$ is close to $0$. 
 \end{proof}

\ignore{
\begin{remark}
Note that in our analysis,  we assume throughout  that the target used in the proof of work function remains  fixed (as in the \cite{backbone}). It is easy to see that if this does not hold then,  the adversary, using selfish mining \cite{selfish},  then the target becomes greater and the difficulty decreases as the total computational power seems smaller than it is. So the adversary can produce easier blocks and maybe has incentives to deviate.
\end{remark}
}
\subsection{Utility Equivalent to Absolute Rewards-Block Reward Changes}
\label{rewardchange1}
We will prove that the previous result holds also in the case when (i) the block reward changes every at least $l\cdot \kappa$ rounds where $l$ a positive constant and $\kappa$ the security parameter and (ii) the environment terminates the execution at least $l \cdot \kappa$ rounds after the last change of the block reward.  \par
By Lemma $\ref{secondlemma}$ we have the following lemma.
\begin{lemma}{\label{min=max2}}
  For every $r$-admissible environment $\mathcal{Z}$ with input $1^{p'(\kappa)}$, where $\kappa$ the security parameter it holds  
\begin{equation*}
R^{\max}_{T}(\mathcal{E}_{\mathcal{Z},H_T})\equiv  R^{\min}_{T}(\mathcal{E}_{\mathcal{Z},H_T})\equiv \sum _{j=1}^{m+2} X_{r_j}^T (\mathcal{E}_{\mathcal{Z},H_T})\cdot w_{j-1}     
\end{equation*} where $r_1,...,r_m$ are the rounds when the block reward changes, $r_0$ is the first round, $r_{m+1}$ the last complete round of execution $\mathcal{E}_{\mathcal{Z},H_T}$, $r_{m+2}=r_{m+1}+1$,  $w_0,w_1,..., w_m=w_{m+1}$ are the block rewards respectively and  $X_{r_j}^T (\mathcal{E}_{\mathcal{Z},H_T})$ are the successful rounds for $T$ between the rounds $r_{j-1}$ and $r_j-1$ included $r_{j-1}$ and $ r_j-1$.
\end{lemma}
Note that $X_{r_j}^T (\mathcal{E}_{\mathcal{Z},H_T})$ is a sum of independent Boolean random variables that are Bernoulli trials.
In addition 
\begin{lemma}{\label{min=max2b}}
\begin{equation*}
R^{\max}_{T}(\mathcal{E'}_{\mathcal{Z},\mathcal{A}})\leq \sum_{j=1}^{m+2}  Z_{r_j} (\mathcal{E'}_{\mathcal{Z},\mathcal{A}})\cdot w_{j-1},
\end{equation*}
 where $r_1,...,r_m$ are the rounds when the block reward changes, $r_0$ is the first round,  $r_{m+1}$ the last complete round of execution   $\mathcal{E}_{\mathcal{Z},H_T}$,  $r_{m+2}= r_{m+1}+1$,  $w_0, w_1,..., w_m=w_{m+1}$ are the block rewards respectively,  $Z_{r_j} (\mathcal{E'}_{\mathcal{Z},\mathcal{A}})$ is the number of blocks produced by the adversary between the rounds $r_{j-1}$ and $r_j-1 $ included $r_{j-1}$ and $ r_j-1$. 
\end{lemma}

\begin{theorem*} 
Supposing that (i) the block reward changes every at least $l\cdot \kappa$ rounds where $l$ a positive constant and $\kappa$ the security parameter and (ii) the environment terminates the execution at least $l \cdot \kappa$ rounds after the last change of the block reward then it holds: for any  $\delta_1 \in(0,0.25)$ such that 
$4 \cdot \delta_1\cdot(1+s)+s<1$, where $s$ the expected number of solutions per round, the Bitcoin with fixed target in a synchronous setting is $(n-1, 4 \cdot \delta_1\cdot(1+s)+s,0)$-EVP according to the utility function absolute rewards (def.\ref{utilities}).
\end{theorem*}

\begin{proof}

We have for $j\in \lbrace 1,...,m\rbrace $ for any $\delta_2 \in (0,1)$
\begin{align*}
  X_{r_j}^T (\mathcal{E}_{\mathcal{Z},H_T})\cdot w_j > \dfrac{p\cdot q \cdot t'}{1+p \cdot q \cdot t'}\cdot ( r_j-r_{j-1})\cdot (1-\delta_2)\cdot w_j> 0
\end{align*}
with overwhelming probability in $r_j-r_{j-1}$ and as $r_j-r_{j-1} \geq l \cdot \kappa$ also in $\kappa$.\par 
In addition 
\begin{align*}
( X_{r_{m+1}}^T (\mathcal{E}_{\mathcal{Z},H_T})+X_{r_{m+2}}^T (\mathcal{E}_{\mathcal{Z},H_T}))\cdot w_m > \dfrac{p\cdot q \cdot t'}{1+p \cdot q \cdot t'}\cdot ( r_{m+1}-r_{m} +1)\cdot (1-\delta_2)\cdot w_m
\end{align*}
with overwhelming probability in $r_{m+1}-r_{m}+1 \geq l \cdot \kappa$.\par
Moreover for $j\in \lbrace 1,...,m\rbrace $  
for any $\delta_1 \in (0,1)$ it holds
\begin{align*}
Z_{r_j}(\mathcal{E'}_{\mathcal{Z},\mathcal{A}})< p\cdot q\cdot t'\cdot ( r_j-r_{j-1})\cdot (1+\delta_1) 
\end{align*}
with overwhelming probability in $r_j-r_{j-1}$ and as $r_j-r_{j-1}\geq l \cdot \kappa$ also in $\kappa$. \par
Also 

\begin{align*}
Z_{r_{m+1}}(\mathcal{E'}_{\mathcal{Z},\mathcal{A}})+Z_{r_{m+2}}(\mathcal{E'}_{\mathcal{Z},\mathcal{A}})< p\cdot q\cdot t'\cdot ( r_{m+1}-r_{m}+1)\cdot (1+\delta_1) 
\end{align*}
with overwhelming probability in $r_{m+1}-r_{m}+1 \geq l\cdot \kappa$. \par
So
for $j\in \lbrace 1,...,m\rbrace $ it holds for any $\delta_1 \in (0,0.25)$
\begin{align*}
Z_{r_j}(\mathcal{E'}_{\mathcal{Z},\mathcal{A}})\cdot w_j < X_{r_j}^T (\mathcal{E}_{\mathcal{Z},H_T})\cdot w_j\cdot( 1+4\cdot \delta_1 +(1+4\cdot\delta_1) \cdot s)
\end{align*}
with overwhelming probability in $\kappa$. \par
In addition 
\begin{align*}
(Z_{r_{m+1}}(\mathcal{E'}_{\mathcal{Z},\mathcal{A}})+Z_{r_{m+2}}(\mathcal{E'}_{\mathcal{Z},\mathcal{A}}) )\cdot w_m  <(X_{r_{m+1}}^T (\mathcal{E}_{\mathcal{Z},H_T})+X_{r_{m+2}}^T (\mathcal{E}_{\mathcal{Z},H_T}))\cdot w_m\cdot l
\end{align*}
with overwhelming probability in $\kappa$, where $l=( 1+4\cdot \delta_1 +(1+4\cdot\delta_1)\cdot s)$.

As a result with overwhelming probability in $\kappa$ it holds for any $\delta_1 \in (0,0.25)$
\begin{align*}
&R^{\max}_{T}(\mathcal{E'}_{\mathcal{Z},\mathcal{A}})\\&\leq \sum_{j=1}^{m} [ Z_{r_j} (\mathcal{E'}_{\mathcal{Z},\mathcal{A}})\cdot w_{j-1} ] + (Z_{r_{m+1}}(\mathcal{E'}_{\mathcal{Z},\mathcal{A}})+Z_{r_{m+2}}(\mathcal{E'}_{\mathcal{Z},\mathcal{A}}) )\cdot w_m\\ &< \sum_{j=1}^{m}[X_{r_j}^T (\mathcal{E}_{\mathcal{Z},H_T})\cdot w_{j-1}\cdot (1+4\cdot \delta_1 +(1+4\cdot\delta_1)\cdot s)]\\ &+ (X_{r_{m+1}}^T (\mathcal{E}_{\mathcal{Z},H_T})+ X_{r_{m+2}}^T (\mathcal{E}_{\mathcal{Z},H_T}))\cdot w_m \cdot (1+4\cdot \delta_1 +(1+4\cdot\delta_1)\cdot s)\\& \overset{w_m=w_{m+1}}\leq  (\sum_{j=1}^{m+2}X_{r_j}^T (\mathcal{E}_{\mathcal{Z},H_T})\cdot w_{j-1})\cdot (1+4\cdot \delta_1 +(1+4\cdot\delta_1)\cdot s)\\& \equiv  R^{\min}_{T}(\mathcal{E}_{\mathcal{Z},H_T}) \cdot  (1+4\cdot \delta_1 +(1+4\cdot\delta_1)\cdot s)
\end{align*}
\end{proof}

\subsection{Absolute Rewards Minus Absolute Cost }
In this subsection we prove that if a static adversary with fixed cost wants to maximize its profit and the cost of each query to the random oracle is small enough compared to the block reward, then the adversary has no incentives to deviate from the Bitcoin protocol even if it controls $n-1$ participants. \par

When we say profit we mean absolute rewards minus absolute cost or in other words the flat reward the adversary gets from the blocks that it has produced and are included in the public ledger
 minus the cost that it has paid due to the queries to the random oracle. 
 \par  Note that the smaller the cost of each query is, the better EVP we have. 
We suppose in this subsection for simplicity that the reward of each block is fixed and equal to $w$. However the theorem also holds when we assume that (i) the block reward changes every at least $l\cdot \kappa$ rounds where $l$ a positive constant and $\kappa$ the security parameter and (ii) the environment terminates the execution at least $l \cdot \kappa$ rounds after the last change of the block reward. The exact theorems and proofs of this case are given in the next subsection.

\begin{theorem*}
Suppose that there exists $\phi \in (0,1-s)$ such that $c<p\cdot w \cdot \phi/(1+p\cdot q\cdot (n-1))$. Then, supposing that the reward of each block is a constant $w$, it holds: for any $\delta_1 \in(0,0.25)$, such that $c\leq p\cdot w \cdot (1-\delta_1)\cdot \phi/(1+p\cdot q\cdot (n-1))$ and $4 \cdot \delta_1\cdot(1+s)+s<1- \phi$, where $s$ the expected number of solutions per round, the Bitcoin with fixed target in a synchronous setting is $(n-1, (4 \cdot \delta_1\cdot(1+s)+s)/(1- \phi),0)$-EVP according to the utility function absolute rewards minus absolute cost (def.\ref{utilities}).
\end{theorem*}
\begin{proof}
We choose an arbitrary $r$-admissible environment $\mathcal{Z}$ with input $1^{p'(\kappa)}$, where $\kappa$ the security parameter and an arbitrary adversary $\mathcal{A}$ static with fixed cost that is PPT and it controls an arbitrary set $T$ with $t'$ participants  where $t' \in \lbrace 1,...,n-1\rbrace$. The adversary as described above has chosen the number of the queries that each participant controlled by the adversary will not ask during each round. Let $x$ the total number of the queries that all the participants controlled by the adversary will not ask during each round.\par
We will have two executions of the Bitcoin protocol with the same environment, but with different adversary : In the first execution $\mathcal{E}_{\mathcal{Z},H_T}$ the adversary is $H_T$ and in the second execution 
$\mathcal{E'}_{\mathcal{Z},\mathcal{A}}$ the adversary is $\mathcal{A}$.\par
Let  $\phi \in (0,1-s)$ such that $c< p\cdot w \cdot \phi /(1+p\cdot q\cdot (n-1))\leq p\cdot w \cdot \phi /(1+p\cdot q\cdot t')$.
We choose an arbitrary  $\delta_1 \in(0,0.25)$ such that $c\leq p\cdot w \cdot (1-\delta_1)\cdot \phi/(1+p\cdot q\cdot (n-1))$ and  $4\cdot \delta_1 +(1+4\cdot\delta_1)\cdot s<1- \phi$ . Then 
by hypothesis and by the fact that $H_T$ follows the  protocol and asks all the queries that are available to the participants controlled by the adversary during each round  we have that :\par 
\begin{align}{\label{H_Tthird}}
 U^{\min}_T(\mathcal{E}_{\mathcal{Z},H_T})&> \dfrac{p\cdot q \cdot t'}{1+p \cdot q \cdot t'}\cdot r\cdot (1-\delta_1)\cdot w -c\cdot q \cdot t' \cdot r\\ &\geq \dfrac{p\cdot q \cdot t'}{1+p \cdot q \cdot t'}\cdot r\cdot (1-\delta_1)\cdot w -\dfrac{q\cdot t'\cdot r\cdot  p\cdot w\cdot (1-\delta_1)\cdot\phi}{(1+p\cdot q\cdot t')}\\{\label{H_Tthird'}}&=\dfrac{p\cdot q \cdot t'}{1+p \cdot q \cdot t'}\cdot r\cdot (1-\delta_1)\cdot w \cdot (1-\phi)\\&>0
\end{align}
with overwhelming probability in $r$ and also in $\kappa$. 
\par Regarding $\mathcal{E'}_{\mathcal{Z},\mathcal{A}}$, given that during each round 
the adversary asks all the available queries except for the $x$ queries that it has specified in the beginning of the execution, the following holds with overwhelming probability in $\kappa$:
\begin{equation}{\label{H_Tsec}}
U^{\max}_T(\mathcal{E'}_{\mathcal{Z},\mathcal{A}})<p\cdot (q\cdot t'-x)\cdot r\cdot (1+\delta_1)\cdot w- c\cdot (q \cdot t'  -x)\cdot r
\end{equation}

By 
our assumption that $ c\leq p\cdot w\cdot (1-\delta_1)\cdot\phi/(1+p\cdot q\cdot t')$ for $\phi\in(0,1-s)$  we have that
 $f(x)= p\cdot (q\cdot t'-x)\cdot r\cdot (1+\delta_1)\cdot w- c\cdot (q \cdot t'  -x)\cdot r$  for $x\in [0,\infty)$ is decreasing. \par So  it holds with overwhelming probability in $\kappa$ that:
\begin{equation}{\label{badfinal}}
U^{\max}_T(\mathcal{E'}_{\mathcal{Z},\mathcal{A}})< p\cdot q\cdot t'\cdot r\cdot (1+\delta_1)\cdot w- c\cdot q \cdot t' \cdot r 
\end{equation}

As a result it holds  with overwhelming probability in $r$ and in $\kappa$
\begin{align*}
& U^{\max}_T(\mathcal{E'}_{\mathcal{Z},\mathcal{A}})-U^{\min}_T(\mathcal{E}_{\mathcal{Z},H_T})\\ &\overset{(\ref{badfinal}),(\ref{H_Tthird})}< p\cdot q\cdot t'\cdot r\cdot (1+\delta_1)\cdot w -\dfrac{p\cdot q \cdot t'}{1+p \cdot q \cdot t'}\cdot r\cdot (1-\delta_1)\cdot w  \\  &= \dfrac{p\cdot q \cdot t'}{1+p \cdot q \cdot t'}\cdot r\cdot (1-\delta_1)\cdot w \cdot (1-\phi) \cdot (\dfrac{(1+\delta_1)\cdot(1+p\cdot q \cdot t')}{(1-\delta_1)\cdot (1-\phi)}-\dfrac{1}{1-\phi})\\&\overset{\ref{H_Tthird'}}< U^{\min}_T(\mathcal{E}_{\mathcal{Z},H_T})\cdot \dfrac{4\cdot \delta_1 +(1+4\cdot\delta_1)\cdot s}{1-\phi}
\end{align*}

\end {proof}
\subsection{Utility Equivalent to Absolute Rewards Minus Absolute Cost-Block Reward Changes} 
\label{rewardchange2}


We will show that the result of the previous subsection holds also when we assume that (i) the block reward changes every at least $l\cdot \kappa$ rounds where $l$ a positive constant and $\kappa$ the security parameter and (ii) the environment terminates the execution at least $l \cdot \kappa$ rounds after the last change of the block reward. 
By Lemmas $\ref{min=max2}$,$\ref{min=max2b}$ and by the fact that the adversary is static with fixed cost we have the following lemmas
\begin{lemma}{\label{min=max3}}
  For every $r$-admissible environment $\mathcal{Z}$ with input $1^{p'(\kappa)}$, where $\kappa$ the security parameter it holds  
\begin{equation*} U^{\max}_{T}(\mathcal{E}_{\mathcal{Z},H_T})\equiv  U^{\min}_{T}(\mathcal{E}_{\mathcal{Z},H_T})\equiv \sum _{j=1}^{m+2} X_{r_j}^T (\mathcal{E}_{\mathcal{Z},H_T})\cdot w_{j-1}  - \sum_{j=1}^{m+2} \sum_{l:P_l \in T} C_{l,{r_j}}(\mathcal{E}_{\mathcal{Z},H_T})
\end{equation*} where $r_1,...,r_m$ are the rounds when the block reward changes, $r_0$ is the first round,  $r_{m+1}$ the last complete round of execution $\mathcal{E}_{\mathcal{Z},H_T}$, $r_{m+2}=r_{m+1}+1$ ,  $w_0,w_1,..., w_m=w_{m+1}$ are the block rewards respectively, $X_{r_j}^T (\mathcal{E}_{\mathcal{Z},H_T})$ are the successful rounds for $T$ and $\sum_{l:P_l \in T} C_{l,{r_j}}(\mathcal{E}_{\mathcal{Z},H_T})$ the cost for $T$ respectively  between the rounds $r_{j-1}$ and $r_j-1$ included $r_{j-1}$ and $ r_j-1$.
\end{lemma}

Recall that the cost of each round is fixed and determined in the beginning of the execution.
\begin{lemma}
\begin{align}
U^{\max}_{T}(\mathcal{E'}_{\mathcal{Z},\mathcal{A}})& \leq \sum_{j=1}^{m+2}  Z_{r_j} (\mathcal{E'}_{\mathcal{Z},\mathcal{A}})\cdot w_{j-1}- \sum_{j=1}^{m+2} \sum_{l:P_l \in T} C_{l,{r_j}}(\mathcal{E'}_{\mathcal{Z},\mathcal{A}})
\end{align}
 where $r_1,...,r_m$ are the rounds when the block reward changes, $r_0$ is the first round, $r_{m+1}$ the last complete round of execution $\mathcal{E'}_{\mathcal{Z},\mathcal{A}}$, $r_{m+2}= r_{m+1}+1$ ,  $w_0, w_1,..., w_m=w_{m+1}$ are the block rewards respectively,  $Z_{r_j} (\mathcal{E'}_{\mathcal{Z},\mathcal{A}})$ is the number of blocks produced by the adversary between the rounds $r_{j-1}$ and $r_j-1 $ included $r_{j-1}$ and $ r_j-1$. 
\end{lemma}

\begin{theorem*}
We assume that (i) the block reward changes every at least $l\cdot \kappa$ rounds where $l$ a positive constant and $\kappa$ the security parameter and (ii) the environment terminates the execution at least $l \cdot \kappa$ rounds after the last change of the block reward.
Let $w_j$ for $ j \in \lbrace 0,...,m \rbrace$ be all the block rewards respectively. Assuming that there exists   $\phi \in (0,1-s)$ such that $c<p\cdot w_j \cdot\phi/(1+p\cdot q\cdot (n-1))$ for all $ j \in \lbrace 0,...,m \rbrace$, then it holds: for any $\delta_1 \in(0,0.25)$, such that $c\leq p\cdot w_j \cdot (1-\delta_1)\cdot \phi/(1+p\cdot q\cdot (n-1))$ for all  $ j \in \lbrace 0,...,m \rbrace$ and $4 \cdot \delta_1\cdot(1+s)+s<1 - \phi$, where $s$ the expected number of solutions per round, the Bitcoin with fixed target in a synchronous setting is  $(n-1, (4 \cdot \delta_1\cdot(1+s)+s)/(1- \phi),0)$-EVP according to the utility function absolute rewards minus absolute cost (def.\ref{utilities}).
\end{theorem*}

\begin{proof}
Let $t'\in\lbrace 1,..., n-1 \rbrace$, $\phi \in (0,1-s)$  such that $c<p\cdot w_j \cdot \phi/(1+p\cdot q\cdot (n-1))\leq p\cdot w_j \cdot \phi/(1+p\cdot q\cdot t')$ for all $ j \in \lbrace 0,...,m \rbrace$ and arbitrary $\delta_1 \in(0,0.25)$  such that $c\leq p\cdot w_j \cdot (1-\delta_1)\cdot \phi/(1+p\cdot q\cdot (n-1))$ for all  $ j \in \lbrace 0,...,m \rbrace$ and $4\cdot \delta_1 +(1+4\cdot\delta_1)\cdot s<1 - \phi$.

By the previous lemmas, by the  assumption that for any $ j\in \lbrace 0,...,m \rbrace $, $c\leq p\cdot w_{j} \cdot (1-\delta_1)\cdot \phi/(1+p\cdot q\cdot t')$ and by Chernoff bound we have  with overwhelming probability in $\kappa$ :

\begin{align*}
U^{\min}_{T}(\mathcal{E}_{\mathcal{Z},H_T}) &\equiv \sum _{j=1}^{m}[ X_{r_j}^T (\mathcal{E}_{\mathcal{Z},H_T})\cdot w_{j-1}] +(X_{r_{m+1}}^T (\mathcal{E}_{\mathcal{Z},H_T})+X_{r_{m+2}}^T (\mathcal{E}_{\mathcal{Z},H_T}))\cdot w_m  - \\& \sum _{j=1}^{m+2} [\sum_{l:P_l \in T} C_{l,{r_j}}(\mathcal{E'}_{\mathcal{Z},\mathcal{A}})] \\& > \sum _{j=1}^{m}  [\dfrac{p\cdot q \cdot t'}{1+p \cdot q \cdot t'}\cdot (r_j-r_{j-1})\cdot (1-\delta_1)\cdot w_{j-1}\cdot (1-\phi)]+ \\& \dfrac{p\cdot q \cdot t'}{1+p \cdot q \cdot t'}\cdot (r_{m+1}-r_{m} +1)\cdot (1-\delta_1)\cdot w_{m}\cdot (1-\phi)\\&>0   
\end{align*}

and 
\begin{align*}
U^{\max}_{T}(\mathcal{E'}_{\mathcal{Z},\mathcal{A}})& \leq \sum_{j=1}^{m} [ Z_{r_j} (\mathcal{E'}_{\mathcal{Z},\mathcal{A}})\cdot w_{j-1}] + (Z_{r_{m+1}} (\mathcal{E'}_{\mathcal{Z},\mathcal{A}})+Z_{r_{m+2}} (\mathcal{E'}_{\mathcal{Z},\mathcal{A}}))\cdot w_{m} -\\& \sum_{j=1}^{m+2} \sum_{l:P_l \in T} C_{l,{r_j}}(\mathcal{E'}_{\mathcal{Z},\mathcal{A}}) \\&< \sum_{j=1}^{m}[ p\cdot (q\cdot t')\cdot (r_j-r_{j-1})\cdot (1+\delta_1)\cdot w_{j-1}- c\cdot ( q \cdot t' )\cdot (r_j-r_{j-1})] + \\& p\cdot (q\cdot t')\cdot (r_{m+1}-r_{m} +1)\cdot (1+\delta_1)\cdot w_{m}- c\cdot ( q \cdot t'  )\cdot (r_{m+1}-r_{m} +1)
\end{align*}
 As a result, we can prove in the same way as the previous subsection that it  holds  with overwhelming probability in $r$ and in $\kappa$ the following : 
\begin{align*}
& U^{\max}_T(\mathcal{E'}_{\mathcal{Z},\mathcal{A}})-U^{\min}_T(\mathcal{E}_{\mathcal{Z},H_T})\leq U^{\min}_T(\mathcal{E}_{\mathcal{Z},H_T})\cdot \dfrac{4\cdot \delta_1 +(1+4\cdot\delta_1)\cdot s}{1-\phi}
\end{align*}

\end{proof}

\subsection{Negative Results: Relative Rewards}

In this subsection we prove that when the utility
is based on relative rewards, i.e., the ratio
of rewards of the strategic coalition of the adversary
over the total rewards of all the participants, the 
Bitcoin with fixed target in asynchronous setting is not EVP. In this way we show how our model can be used to prove negative results.
The core idea is to use the selfish mining strategy \cite{selfish,selfish2,selfish3,selfish4,perf}  
to construct an attack that invalidates the equilibrium property. This kind of attack was used also in \cite{backbone} as argument for the tightness of ``chain quality".
Without loss of generality, we will assume that the reward of each block is the same and equal to $w$ (the negative result carries trivially to the general case). 

In more detail, for an arbitrary $t \in \lbrace 1,...,n-1 \rbrace $ and $t'<\min \lbrace n/2,t+1 \rbrace$ we show that the protocol is not a $(t,\epsilon,\epsilon')$-EVP 
for $\epsilon + \epsilon'<\dfrac{t'}{n-t'}\cdot(1-\delta') -\dfrac{t'}{n}\cdot (1+\delta'')\cdot (1+s)$, for  $\delta',\delta''$ small. Recall that $s=p\cdot q \cdot n$ are the expected number of solutions per round. 
\ignore{

we prove the following. We construct
an adversary that commands an arbitrary number of participants $t'<n/2$,
and we show the protocol is not a $(t',\epsilon,\epsilon')$-EVP 
for $\epsilon + \epsilon'<\dfrac{t'}{n-t'}\cdot(1-\delta') -\dfrac{t'}{n}\cdot (1+\delta'')\cdot (1+s)$, for  $\delta',\delta''$ small. Recall that $s$ are the expected number of solutions per round. 

}

\begin{theorem*}
 Let  $t \in \lbrace 1,...,n-1 \rbrace $ and $t'<\min \lbrace n/2,t+1 \rbrace$. Then for any $ \epsilon +\epsilon' <\dfrac{t'}{n-t'}\cdot(1-\delta') -\dfrac{t'}{n}\cdot (1+\delta'')\cdot (1+s)$, where $s$ the expected number of solutions per round, for some $\delta',\delta''$, following the Bitcoin with fixed target in asynchronous setting is not 
$(t,\epsilon, \epsilon')$-EVP  according to the utility function relative rewards (def.\ref{utilities}). 
\end{theorem*}
\begin{proof}
Let $t \in \lbrace 1,...,n-1 \rbrace $. 
We consider an arbitrary $r$-admissible environment $\mathcal{Z}$ with input $1^{p'(\kappa)}$ and we will describe a PPT static  adversary $A_0$ with fixed cost (who controls a set $T$ with $t'<\min \lbrace n/2,t+1 \rbrace$ participants) such that it holds with high probability :
\begin{align}
&U^{\max}_T(\mathcal{E'}_{\mathcal{Z},A_0})-
U^{\min}_T(\mathcal{E}_{\mathcal{Z},H_T} )\\&\geq   \dfrac{   t' \cdot (1-\delta_1)\cdot (1-\epsilon''')} {(n-t')\cdot (1+\delta _4)}-\dfrac{t'\cdot (1+p\cdot q\cdot n)}{n}\cdot\dfrac{(1+\delta_2)}{(1-\delta_3)}{\label{epsilon'}}\\ & = B 
\end{align}
 for  any $ \delta_1,\delta_2,\delta_3,\delta_4\in(0,1) $  and a small  $ \epsilon''' >0$.

\ignore{
it holds with high probability 

$$ U^{\max}_T(\mathcal{E}_{\mathcal{Z},A_0}) \geq  U^{\min}_T(\mathcal{E}_{\mathcal{Z},H_T} )\cdot (1+ \dfrac{t'}{n-t'}\cdot(1-\delta)-\delta)$$ for $\delta$ small.}
After that we prove our theorem by contradiction. 
In more detail, we suppose that there exist $\epsilon,\epsilon'$ such that $\epsilon+\epsilon'<B$ 
so that the Bitcoin with fixed target in asynchronous setting is  $(t,\epsilon,\epsilon')$-EVP and we will end up in contradiction
. In other words we suppose that there exist $\epsilon,\epsilon'$ such that $\epsilon+\epsilon'<B$ so that 
\begin{align}
&U^{\max}_T(\mathcal{E'}_{\mathcal{Z},\mathcal{A}})-
U^{\min}_T(\mathcal{E}_{\mathcal{Z},H_T} )\leq \mid U^{\min}_T(\mathcal{E}_{\mathcal{Z},H_T} )\mid \cdot\epsilon+\epsilon'
\end{align} with overwhelming probability, where $A$ an arbitrary PPT static adversary with fixed cost that controls a set $T$ with at most $t$ participants and $Z$ an arbitrary $r$-admissible environment with input $1^{p'(\kappa)}$. \par  
Then we have 
\begin{align}
&U^{\max}_T(\mathcal{E'}_{\mathcal{Z},\mathcal{A}})-
U^{\min}_T(\mathcal{E}_{\mathcal{Z},H_T} )\\&\leq \mid U^{\min}_T(\mathcal{E}_{\mathcal{Z},H_T} )\mid \cdot\epsilon+\epsilon'\\&\leq \epsilon+\epsilon'\\&<B
\end{align} 
with overwhelming probability. \par 
However this does not hold because there exists $A_0$ that satisfies ($\ref{epsilon'}$).


In order to prove equation (\ref{epsilon'}):
 firstly we find an upper  bound for $ U^{\min}_T(\mathcal{E}_{\mathcal{Z},H_T} )$ that holds with overwhelming probability in the security parameter $\kappa$. 
 
Recall that $T$ is an arbitrary set with $t'<\min \lbrace n/2,t+1 \rbrace$ participants that adversary $A_0$, whom we will describe  later, controls.

 \par 
By Lemma $\ref{secondlemma}$ 
\begin{align*}
R^{\min}_{T}(\mathcal{E}_{\mathcal{Z},H_T})\equiv R^{\max}_{T}(\mathcal{E}_{\mathcal{Z},H_T})\equiv X^T_r(\mathcal{E}_{\mathcal{Z},H_T})\cdot w
\end{align*}
 As a result by Chernoff bound it holds for any $\delta_2\in(0,1)$ with overwhelming probability in $r$ and also in $\kappa$  :
\begin{equation}{\label{third' H_T}}
0<R^{\max}_{T}(\mathcal{E}_{\mathcal{Z},H_T})< p\cdot q\cdot t' \cdot(1+\delta_2)\cdot w \cdot r.
\end{equation}

By Lemma  $\ref{firstlemma}$
\begin{align}
R^{\min}_{S}(\mathcal{E}_{\mathcal{Z},H_T})\equiv R^{\max}_{S}(\mathcal{E}_{\mathcal{Z},H_T}) \equiv X^{S}_r(\mathcal{E}_{\mathcal{Z},H_T})\cdot w
\end{align}

As a result for any $\delta_3\in(0,1)$ with overwhelming probability in $r$ and also in $\kappa$ :
\begin{align}{\label{S,H_T}}
R^{\min}_{S}(\mathcal{E}_{\mathcal{Z},H_T})\equiv R^{\max}_{S}(\mathcal{E}_{\mathcal{Z},H_T})& > \dfrac{p\cdot q \cdot n}{1+p \cdot q \cdot n}\cdot r\cdot (1-\delta_3)\cdot w>0
\end{align} 
Recall that the executions last at least one round.
So we know that with overwhelming probability in $\kappa $ for any $j$ such that $P_j$ honest  

$$U^{j}_{T}(\mathcal{E}_{\mathcal{Z},H_T})\equiv \dfrac{ R^{j}_{T}(\mathcal{E}_{\mathcal{Z},H_T})}{ R^{j}_{S}(\mathcal{E}_{\mathcal{Z},H_T})}$$
as $R^{j}_{S}(\mathcal{E}_{\mathcal{Z},H_T})\neq 0$.

As a result we have 
that for any $\delta_2,\delta_3\in(0,1)$ it holds with overwhelming probability in $r$ and also in $\kappa$ that:
\begin{align} {\label{BASIC}}
U^{\min}_T(\mathcal{E}_{\mathcal{Z},H_T} )&\leq U^{\max}_T(\mathcal{E}_{\mathcal{Z},H_T} )\\ &\leq \dfrac{ R^{\max}_{T}(\mathcal{E}_{\mathcal{Z},H_T})} {R^{\min}_{S}(\mathcal{E}_{\mathcal{Z},H_T})}\\ & \overset{ (\ref{third' H_T}),(\ref{S,H_T})}\leq
\dfrac{p\cdot q\cdot t' \cdot(1+\delta_2)\cdot w \cdot r } {\dfrac{p\cdot q \cdot n}{1+p \cdot q \cdot n}\cdot r\cdot (1-\delta_3)\cdot w}\\&=\dfrac{t'\cdot (1+p\cdot q\cdot n)}{n}\cdot\dfrac{(1+\delta_2)}{(1-\delta_3)}
\end{align}

 Now we will describe the adversary $A_0$ who does a type of selfish mining, \cite{selfish,selfish2,selfish3,selfish4,perf}, which was described also in \cite{backbone} and we will find a lower  bound  for $U^{\max}_{T}(\mathcal{E}_{\mathcal{Z},A_0})$ with high probability (not negligible). \par 
 $A_0$ chooses to ask all the queries. Initially extends the chain coming from an honest participant, but when it finds a block it does not send it to the Diffuse Functionality. It continues working on its private chain until another participant announces a block. Then the adversary reveals one of its blocks to all the honest participants.  When this happens  all the honest participants adopt its block instead of the block coming from the honest participant. If the adversarial private chain becomes smaller than the chain coming from an honest participant then the adversary adopts the honest participant's chain. Note that  when one of the participants controlled by the adversary finds a block during a round, the adversary uses the rest available queries for finding a block that extends this block. 
\begin{itemize}
\item  $X^{S\setminus T}_r(\mathcal{E'}_{\mathcal{Z},A_0})\equiv \sum_{m=1}^{r}X^{S\setminus T,m}(\mathcal{E'}_{\mathcal{Z},A_0})$ is the number of the successful rounds for $S\setminus T$ until the last complete round $r$ of $\mathcal{E'}_{\mathcal{Z},A_0}$.
\item  $Z_r(\mathcal{E'}_{\mathcal{Z},A_0})\equiv \sum_{i=1}^{r}\sum_{k=1}^{t'}\sum_{j=1}^{q-x_k} Z_{i,j,k}(\mathcal{E'}_{\mathcal{Z},A_0})$. $Z_r(\mathcal{E'}_{\mathcal{Z},A_0})$  is the number of the blocks the adversary $A_0$ has produced until the last complete round $r$ of $\mathcal{E'}_{\mathcal{Z},A_0}$.
 \end{itemize}

We have that
\begin{align}
R^{\max}_S(\mathcal{E'}_{\mathcal{Z},A_0})\equiv R^{\min}_S(\mathcal{E'}_{\mathcal{Z},A_0})\equiv X^{S\setminus T}_r(\mathcal{E'}_{\mathcal{Z},A_0})\cdot w 
\end{align}
 
This holds because of Lemma $\ref{firstlemma}$ and due to the fact that the adversary $A_0$ does not contribute to the extension of the public ledger  as it only replaces blocks. In addition it announces its blocks to all honest participants.

In addition with overwhelming probability in $\kappa$ by Chernoff bound  $$R^{\min}_S(\mathcal{E'}_{\mathcal{Z},A_0})\equiv R^{\max}_S(\mathcal{E'}_{\mathcal{Z},A_0})>0$$ and as a result with overwhelming probability in $\kappa$
$$U^{j}_{T}(\mathcal{E}_{\mathcal{Z},A_0})=  \dfrac{ R^{j}_{T}(\mathcal{E}_{\mathcal{Z},A_0})}{ R^{j}_{S}(\mathcal{E}_{\mathcal{Z},A_0})}$$
\ignore{
By \cite{backbone} we have:
\begin{lemma}{\label{lemma3}}
\begin{align}
R^{\max}_T(\mathcal{E'}_{\mathcal{Z},A_0})\equiv R^{\min}_T(\mathcal{E'}_{\mathcal{Z},A_0})\equiv   Z_r(\mathcal{E'}_{\mathcal{Z},A_0})\cdot (1-\epsilon''')\cdot  w 
\end{align}
 with high probability for small  $ \epsilon''' >0$.    
\end{lemma}
}
Regarding $ R^{\min}_T(\mathcal{E'}_{\mathcal{Z},A_0})$: the adversary $A_0$ announces its block only if an honest participant finds a block and when this happens, it announces it to all the honest participants. The honest participants always adopt its blocks . So $R^{\max}_T(\mathcal{E'}_{\mathcal{Z},A_0})\equiv R^{\min}_T(\mathcal{E'}_{\mathcal{Z},A_0})$ with probability $1$. The number of the adversarial blocks $B(\mathcal{E'}_{\mathcal{Z},A_0})$ in the local chain of an arbitrary honest at the end of the last complete round $r$ of the execution  $\mathcal{E'}_{\mathcal{Z},A_0}$  are, as stated in \cite{backbone}, with high probability equal to the number of the blocks $Z_r(\mathcal{E'}_{\mathcal{Z},A_0})$ produced by the adversary minus a quantity bounded by  $ \epsilon ''\cdot p \cdot q\cdot r \cdot t'$, for small  $ \epsilon'' >0$. \par This happens because when the adversary $A_0$ has found more than one block during each round it means that all these blocks form a chain and extend the length of the local chains of all the honest participants. Note that when the execution ends, the adversary may have a small quantity of blocks that are unused in the case the honest participants did not have enough successful rounds.
\par 
Recall that contrary to the adversary, when the honest participants have found more than one block during a round, these blocks do not form a chain, because (i) an honest participant never sends more than one block to the Diffuse Functionality, and (ii) when an honest participant receives a block from another participant, it 
does not extend this new block until the end of the round.

By Chernoff bound 
  we have with high probability  for  any $\delta_1\in(0,1)$ and 
a small  $ \epsilon''' >0$  
\begin{equation}
R^{\min}_T(\mathcal{E'}_{\mathcal{Z},A_0}) \geq p \cdot q \cdot t' \cdot r \cdot (1-\delta_1)\cdot w \cdot (1-\epsilon''')
 \end{equation}
 
 Moreover by Chernoff  bound  it holds with overwhelming probability in $\kappa$  for  any $ \delta_4\in(0,1) $ 
\begin{equation}
R^{\max}_S(\mathcal{E'}_{\mathcal{Z},A_0})\leq p\cdot q \cdot (n-t')\cdot r \cdot w \cdot (1+\delta _4)
\end{equation}

So  we have with high probability  for  any $ \delta_4,\delta_1\in(0,1) $ and small  $ \epsilon''' >0$

 \begin{align}
 U^{\max}_T(\mathcal{E'}_{\mathcal{Z},A_0})& \geq  U^{\min}_T(\mathcal{E'}_{\mathcal{Z},A_0})\\& \geq \dfrac{ R^{\min}_{T}(\mathcal{E'}_{\mathcal{Z},A_0})} {R^{\max}_{S}(\mathcal{E'}_{\mathcal{Z},A_0})}\\ & \geq  \dfrac{  p \cdot q \cdot t' \cdot r \cdot (1-\delta_1)\cdot (1-\epsilon''')\cdot w} {p\cdot q \cdot (n-t')\cdot r \cdot w \cdot (1+\delta _4) }\\ &= \dfrac{   t'  \cdot (1-\delta_1) \cdot (1-\epsilon''')} {(n-t')\cdot (1+\delta _4)}
 \end{align}

Finally by the above and by equality ($ \ref{BASIC}$) for  any $ \delta_1,\delta_2,\delta_4,\delta_3\in(0,1) $  and small  $ \epsilon''' >0$ it holds with high probability :
\begin{align}
&U^{\max}_T(\mathcal{E'}_{\mathcal{Z},A_0})-
U^{\min}_T(\mathcal{E}_{\mathcal{Z},H_T} )\\&\geq   \dfrac{   t' \cdot (1-\delta_1)\cdot (1-\epsilon''')} {(n-t')\cdot (1+\delta _4)}-\dfrac{t'\cdot (1+p\cdot q\cdot n)}{n}\cdot\dfrac{(1+\delta_2)}{(1-\delta_3)}\\ & = B 
\end{align}
\end{proof}

\section{Our Proofs Regarding Incentives in a Fair Blockchain Protocol and the Fruitchain Protocol}\label{fruitchain}

\subsection{Proof of Theorem \ref{theorem8}}\label{fruitproof1}
\begin{proof}
 We choose an arbitrary $r$-admissible environment $\mathcal{Z}$ with input $1^{p'(\kappa)}$, where $\kappa$ the security parameter and an arbitrary adversary $\mathcal{A}$ static  that is PPT and it controls a set $T$ that it includes $t'\leq t$ participants .  
We will examine two executions of the blockchain  protocol with the same environment $\mathcal{Z}$, but with different adversary : In the first execution $\mathcal{E}_{\mathcal{Z},H_T}$ the adversary is $H_T$ and in the second execution 
$\mathcal{E'}_{\mathcal{Z},\mathcal{A}}$ the adversary is $\mathcal{A}$. \par We will prove that  with overwhelming probability in the security parameter for any $j: P_j \quad \textrm{honest}$  we have:
\begin{equation}{\label{fruita}}
U^{j}_T(\mathcal{E'}_{\mathcal{Z},\mathcal{A}})\equiv \dfrac{R^{j}_{T}(\mathcal{E'}_{\mathcal{Z},\mathcal{A}})}{R^{j}_{S}(\mathcal{E'}_{\mathcal{Z},\mathcal{A}})}  \leq \dfrac{\sum_{l:P_l\in T}C_{l}(\mathcal{E}_{\mathcal{Z},H_T})}{\sum_{l:P_l\in S} C_{l}(\mathcal{E}_{\mathcal{Z},H_T})} +\delta\cdot \dfrac{\sum_{l:P_l\in S\setminus T}C_{l}(\mathcal{E}_{\mathcal{Z},H_T})}{\sum_{l:P_l\in S} C_{l}(\mathcal{E}_{\mathcal{Z},H_T})}
\end{equation}
By $(t,\delta)$-weak fairness and by the fact that for any $j: P_j  \textrm{ honest}$ it holds with overwhelming probability  $R^{j}_{S}(\mathcal{E'}_{\mathcal{Z},\mathcal{A}})>0$ we have the following result:\par 
 for any $j: P_j\textrm{ honest}$ it holds with overwhelming probability in the security parameter
\begin{align*}
&R^{j}_{S\setminus T}(\mathcal{E'}_{\mathcal{Z},\mathcal{A}})\geq (1-\delta)\cdot \dfrac{\sum_{l:P_l\in S\setminus T}C_{l}(\mathcal{E}_{\mathcal{Z},H_T})}{\sum_{l:P_l\in S} C_{l}(\mathcal{E}_{\mathcal{Z},H_T})}\cdot R^{j}_{S}(\mathcal{E'}_{\mathcal{Z},\mathcal{A}})\Rightarrow\\ &R^{j}_{T}(\mathcal{E'}_{\mathcal{Z},\mathcal{A}})\leq R^{j}_{S}(\mathcal{E'}_{\mathcal{Z},\mathcal{A}})\cdot(1- (1-\delta)\cdot \dfrac{\sum_{l:P_l\in S\setminus T}C_{l}(\mathcal{E}_{\mathcal{Z},H_T})}{\sum_{l:P_l\in S} C_{l}(\mathcal{E}_{\mathcal{Z},H_T})})\Rightarrow\\& R^{j}_{T}(\mathcal{E'}_{\mathcal{Z},\mathcal{A}})\leq R^{j}_{S}(\mathcal{E'}_{\mathcal{Z},\mathcal{A}})\cdot(\dfrac{\sum_{l:P_l\in S}C_{l}(\mathcal{E}_{\mathcal{Z},H_T})}{\sum_{l:P_l\in S} C_{l}(\mathcal{E}_{\mathcal{Z},H_T})}- (1-\delta)\cdot \dfrac{\sum_{l:P_l\in S\setminus T}C_{l}(\mathcal{E}_{\mathcal{Z},H_T})}{\sum_{l:P_l\in S} C_{l}(\mathcal{E}_{\mathcal{Z},H_T})})\Rightarrow \\ &\dfrac{R^{j}_{T}(\mathcal{E'}_{\mathcal{Z},\mathcal{A}})}{R^{j}_{S}(\mathcal{E'}_{\mathcal{Z},\mathcal{A}})}\leq \dfrac{\sum_{l:P_l\in T}C_{l}(\mathcal{E}_{\mathcal{Z},H_T})}{\sum_{l:P_l\in S} C_{l}(\mathcal{E}_{\mathcal{Z},H_T})} +\delta\cdot \dfrac{\sum_{l:P_l\in S\setminus T}C_{l}(\mathcal{E}_{\mathcal{Z},H_T})}{\sum_{l:P_l\in S} C_{l}(\mathcal{E}_{\mathcal{Z},H_T})}
\end{align*}
Note that with overwhelming probability $R^{j}_{S}(\mathcal{E'}_{\mathcal{Z},\mathcal{A}})>0$ and as a result 
\begin{align}
U^{j}_T(\mathcal{E'}_{\mathcal{Z},\mathcal{A}})\equiv \dfrac{R^{j}_{T}(\mathcal{E'}_{\mathcal{Z},\mathcal{A}})}{R^{j}_{S}(\mathcal{E'}_{\mathcal{Z},\mathcal{A}})}
\end{align}
 \par 
By weak fairness and by the fact that it holds with overwhelming probability $R^{\min}_{S}(\mathcal{E}_{\mathcal{Z},\mathcal{H_T}})>0$  we have the following result: \par
\begin{equation}{\label{fruit2}}
U^{min}_T(\mathcal{E}_{\mathcal{Z},H_T})\geq (1 -\delta)\cdot \dfrac{\sum_{l:P_l\in T}C_{l}(\mathcal{E}_{\mathcal{Z},H_T})}{\sum_{l:P_l\in S} C_{l}(\mathcal{E}_{\mathcal{Z},H_T})}
\end{equation}

By equations $(\ref{fruita}),(\ref{fruit2})$  we have that with overwhelming probability in the security parameter  
\begin{align*}
U^{\max}_T(\mathcal{E'}_{\mathcal{Z},\mathcal{A}})-U^{min}_T(\mathcal{E}_{\mathcal{Z},H_T})&\leq \dfrac{\sum_{l:P_l\in T}C_{l}(\mathcal{E}_{\mathcal{Z},H_T})}{\sum_{l:P_l\in S} C_{l}(\mathcal{E}_{\mathcal{Z},H_T})} + \\ &\delta\cdot \dfrac{\sum_{l:P_l\in S\setminus T}C_{l}(\mathcal{E}_{\mathcal{Z},H_T})}{\sum_{l:P_l\in S} C_{l}(\mathcal{E}_{\mathcal{Z},H_T})}-(1 -\delta)\cdot \dfrac{\sum_{l:P_l\in T}C_{l}(\mathcal{E}_{\mathcal{Z},H_T})}{\sum_{l:P_l\in S} C_{l}(\mathcal{E}_{\mathcal{Z},H_T})}\\&\leq \delta
\end{align*}

\end{proof}
\subsection{Proof of Theorem \ref{relativefruitchain}}
\label{fruitproof2}
\begin{proof}
Given that the Fruitchain protocol satisfies $(T_0, \delta)$-approximate fairness property when the adversary controls at most $n/2-1 $ participants, then it satisfies also $(n/2-1, \delta)$-weak fairness property under the restriction that the environment performs the protocol so many rounds that with overwhelming probability (in the security parameter) any honest participant has a chain of at least $T_0$ fruits. Note that by chain growth rate proved in \cite{fruitchain} when 
$r \geq \frac{T_0}{p_f \cdot (\frac{n}{2}+1) \cdot (1- \delta)\cdot q}$ and the adversary controls at most $n/2-1 $ participants, then indeed it holds that with overwhelming probability any honest participant has a chain of at least $T_0$ fruits. 
In addition, by Chernoff bound and by the fact that the execution lasts at least one round, it holds with overwhelming probability in $\kappa$ the following:
for any $j: P_j \textrm{ honest}$, for any PPT static  adversary $\mathcal{A}$ that controls at most $n/2-1$ participants and for any $r$-admissible environment $\mathcal{Z}$ with input $1^{p'(\kappa)}$ $R^{j}_{S}(\mathcal{E'}_{\mathcal{Z},\mathcal{A}})> 0$. So by Theorem \ref{theorem8} we have that the Fruitchain protocol is  $( n/2-1,0,\delta)$-EVP under an $r$-admissible environment where $r \geq \frac{T_0}{p_f \cdot (\frac{n}{2}+1) \cdot (1- \delta)\cdot q}$. 
\end{proof}

\subsection{Proof of Theorem \ref{absolutefruitchain}}
\label{fruitproof3}
\par In this setting the adversary  again is PPT, static with fixed cost, it controls a set of participants $T=\lbrace  P_{i_1},...,P_{i_{t'}}\rbrace \subseteq \lbrace P_1,...,P_n \rbrace=S$  and chooses in the beginning the number $x_m$ of the questions that each participant controlled by the adversary $P_{i_m}$ will not ask during  each round of the execution. 
\begin{proof}
 Let an arbitrary $\delta_1 \in (0,0.25)$ such that $c\leq p_f\cdot w_f \cdot (1-\delta_1)\cdot \phi$  and  $4 \cdot \delta_1<1- \phi$.
We choose also an arbitrary $r$-admissible environment $\mathcal{Z}$ with input $1^{p'(\kappa)}$, where $\kappa$ the security parameter and an arbitrary adversary $\mathcal{A}$ static with fixed cost that is PPT and it has corrupted a set $T$ with $t'$ participants, where $t' \in \lbrace 1,...,n-1\rbrace$. Note that if the adversary controls zero participants then the proof is trivial because adversary's utility is always zero. Let $x=\sum_{m=1}^{t'} x_m$ be the total number of the queries that all the corrupted participants collectively do not ask during each round. Note that $x$ is a constant, not a random variable, as it is determined in the beginning by the static adversary. It holds $0\leq x \leq q \cdot t' $.\par
We will examine two executions of the Fruitchain  protocol with the same environment, but with different adversary: in the first execution $\mathcal{E}_{\mathcal{Z},H_T}$ the adversary is $H_T$ and in the second execution 
$\mathcal{E'}_{\mathcal{Z},\mathcal{A}}$ the adversary is $\mathcal{A}$. Note that the last complete round of the executions is $r$.\par 
Firstly we have:

\begin{equation}
 U^{min}_T(\mathcal{E}_{\mathcal{Z},H_T}) \geq  q \cdot t'\cdot p_f \cdot r \cdot(1-\delta_1)\cdot w_f - c\cdot q\cdot t' \cdot r \geq   q \cdot t'\cdot p_f \cdot r \cdot(1-\delta_1)\cdot w_f \cdot (1- \phi) >0
\end{equation}
with overwhelming probability in $\kappa$.\par 
The above equation is proved by Chernoff bound and taking into account that all the fruits produced by $T$ will be included in the local chain of all the honest participants at the end of the  round $r$.\par 
In addition, the adversary cannot earn rewards for more fruits than that it has produced. Moreover $c\leq p_f\cdot w_f\cdot (1-\delta_1)\cdot \phi$ . As a result by Chernoff bound

\begin{equation}
 U^{\max}_T(\mathcal{E'}_{\mathcal{Z},\mathcal{A}}) \leq  (q \cdot t'-x)\cdot p_f \cdot r \cdot(1+\delta_1)\cdot w_f - c\cdot (q\cdot t'-x)  \cdot r\leq 
 q \cdot t'\cdot p_f \cdot r \cdot(1+\delta_1)\cdot w_f - c\cdot q\cdot t'  \cdot r
\end{equation}
with overwhelming probability in $\kappa$.\par 
As a result  
\begin{align}
 U^{\max}_T(\mathcal{E'}_{\mathcal{Z},\mathcal{A}}) -  U^{min}_T(\mathcal{E}_{\mathcal{Z},H_T}) &\leq (\dfrac{1+\delta_1}{1-\delta_1} -1)\cdot \dfrac{1}{1- \phi}\cdot  U^{\min}_T(\mathcal{E}_{\mathcal{Z},H_T})\\&\leq 4\cdot \delta_1\cdot \dfrac{1}{1- \phi}\cdot  U^{\min}_T(\mathcal{E}_{\mathcal{Z},H_T})
\end{align}
with overwhelming probability in $\kappa$. 
\end{proof}

%% file: arxivdec19.bbl
\begin{thebibliography}{10}

\bibitem{cloud}
How does cloud mining bitcoin work?
\newblock
  \url{https://www.coindesk.com/information/cloud-mining-bitcoin-guide}.

\bibitem{multi}
{\em Algorithmic Game Theory}.
\newblock Cambridge University Press, 2007.
\newblock Nisan, N., Roughgarden, T., Tardos, E., \& Vazirani, V. (Eds.).
  doi:10.1017/CBO9780511800481.

\bibitem{gc}
Ittai Abraham, Lorenzo Alvisi, and Joseph~Y. Halpern.
\newblock Distributed computing meets game theory: Combining insights from two
  fields.
\newblock {\em SIGACT News}, 42(2):69--76, June 2011.

\bibitem{defin}
Ittai Abraham, Danny Dolev, Rica Gonen, and Joe Halpern.
\newblock Distributed computing meets game theory: Robust mechanisms for
  rational secret sharing and multiparty computation.
\newblock In {\em Proceedings of the Twenty-fifth Annual ACM Symposium on
  Principles of Distributed Computing}, PODC '06, pages 53--62, New York, NY,
  USA, 2006. ACM.

\bibitem{Abraham:2008:ATP:1400751.1400804}
Ittai Abraham, Danny Dolev, and Joseph~Y. Halpern.
\newblock An almost-surely terminating polynomial protocol for asynchronous
  byzantine agreement with optimal resilience.
\newblock In {\em Proceedings of the Twenty-seventh ACM Symposium on Principles
  of Distributed Computing}, PODC '08, pages 405--414, New York, NY, USA, 2008.
  ACM.

\bibitem{10.1007/978-3-642-41527-2_5}
Ittai Abraham, Danny Dolev, and Joseph~Y. Halpern.
\newblock Distributed protocols for leader election: A game-theoretic
  perspective.
\newblock In Yehuda Afek, editor, {\em Distributed Computing}, pages 61--75,
  Berlin, Heidelberg, 2013. Springer Berlin Heidelberg.

\bibitem{sol}
Ittai Abraham, Dahlia Malkhi, Kartik Nayak, Ling Ren, and Alexander Spiegelman.
\newblock Solidus: An incentive-compatible cryptocurrency based on
  permissionless byzantine consensus.
\newblock {\em CoRR}, abs/1612.02916, 2016.

\bibitem{Aiyer:2005:BFT:1095810.1095816}
Amitanand~S. Aiyer, Lorenzo Alvisi, Allen Clement, Mike Dahlin, Jean-Philippe
  Martin, and Carl Porth.
\newblock Bar fault tolerance for cooperative services.
\newblock In {\em Proceedings of the Twentieth ACM Symposium on Operating
  Systems Principles}, SOSP '05, pages 45--58, New York, NY, USA, 2005. ACM.

\bibitem{strongn}
Robert~J. Aumann.
\newblock {\em Acceptable Points in General Cooperative n-Person Games.
  Contributions to the Theory of Games (AM-40)}, volume~4, pages 287--324.
\newblock Albert William Tucker, Robert Duncan Luce, Princeton: Princeton
  University Press, 1959.
\newblock Book DOI: \url{ https://doi.org/10.1515/9781400882168}.

\bibitem{baloon}
Moshe Babaioff, Shahar Dobzinski, Sigal Oren, and Aviv Zohar.
\newblock On bitcoin and red balloons.
\newblock In {\em Proceedings of the 13th ACM Conference on Electronic
  Commerce}, EC '12, pages 56--73, New York, NY, USA, 2012. ACM.

\bibitem{puzzle}
Adam Back.
\newblock Hashcash.
\newblock \url{http://www.cypherspace.org/hashcash}, 1997.

\bibitem{zikas}
Christian Badertscher, Juan Garay, Ueli Maurer, Daniel Tschudi, and Vassilis
  Zikas.
\newblock But why does it work? a rational protocol design treatment of
  bitcoin.
\newblock In Jesper~Buus Nielsen and Vincent Rijmen, editors, {\em Advances in
  Cryptology -- EUROCRYPT 2018}, pages 34--65, Cham, 2018. Springer
  International Publishing.

\bibitem{reasoning}
Suguman Bansal.
\newblock {\em Reasoning about incentive compatibility}.
\newblock POPL 2016 Student Research Competition, 2016.

\bibitem{tor}
Iddo Bentov, Pavel Hub'{a}v{c}ek, Tal Moran, and Asaf Nadler.
\newblock Tortoise and hares consensus: the meshcash framework for
  incentive-compatible, scalable cryptocurrencies.
\newblock {\em {IACR} Cryptology ePrint Archive}, 2017:300, 2017.

\bibitem{snow}
Iddo Bentov, Rafael Pass, and Elaine Shi.
\newblock Snow white: Provably secure proofs of stake.
\newblock {\em {IACR} Cryptology ePrint Archive}, 2016:919, 2016.

\bibitem{coalition}
B.Douglas Bernheim, Bezalel Peleg, and Michael~D Whinston.
\newblock Coalition-proof nash equilibria i. concepts.
\newblock {\em Journal of Economic Theory}, 42(1):1 -- 12, 1987.

\bibitem{latest}
Bruno Biais, Christophe Bisiere, Matthieu Bouvard, and Catherine Casamatta.
\newblock {\em The Blockchain Folk Theorem}.
\newblock Swiss Finance Institute Research Paper No. 17-75, 2018.

\bibitem{bribery}
Joseph Bonneau.
\newblock Why buy when you can rent? - bribery attacks on bitcoin-style
  consensus.
\newblock In Jeremy Clark, Sarah Meiklejohn, Peter Y.~A. Ryan, Dan~S. Wallach,
  Michael Brenner, and Kurt Rohloff, editors, {\em Financial Cryptography and
  Data Security - {FC} 2016 International Workshops, BITCOIN, VOTING, and WAHC,
  Christ Church, Barbados, February 26, 2016, Revised Selected Papers}, volume
  9604 of {\em Lecture Notes in Computer Science}, pages 19--26. Springer,
  2016.

\bibitem{Canetti4}
R.~Canetti.
\newblock Universally composable security: A new paradigm for cryptographic
  protocols.
\newblock In {\em Proceedings of the 42Nd IEEE Symposium on Foundations of
  Computer Science}, FOCS '01, pages 136--145, Washington, DC, USA, 2001. IEEE
  Computer Society.

\bibitem{Canetti2}
Ran Canetti.
\newblock Security and composition of multiparty cryptographic protocols.
\newblock {\em Journal of Cryptology}, 13(1):143--202, Jan 2000.

\bibitem{Canetti1}
Ran Canetti.
\newblock Universally composable security: A new paradigm for cryptographic
  protocols.
\newblock Cryptology ePrint Archive, Report 2000/067, 2000.
\newblock \url{https://eprint.iacr.org/2000/067}.

\bibitem{selfish3}
Miles Carlsten, Harry Kalodner, S.~Matthew Weinberg, and Arvind Narayanan.
\newblock On the instability of bitcoin without the block reward.
\newblock In {\em Proceedings of the 2016 ACM SIGSAC Conference on Computer and
  Communications Security}, CCS '16, pages 154--167, New York, NY, USA, 2016.
  ACM.

\bibitem{cheap1}
Vincent~P. Crawford and Joel Sobel.
\newblock Strategic information transmission.
\newblock {\em Econometrica}, 50(6):1431--1451, 1982.

\bibitem{DBLP:conf/podc/DaniMRS11}
Varsha Dani, Mahnush Movahedi, Yamel Rodriguez, and Jared Saia.
\newblock Scalable rational secret sharing.
\newblock In Cyril Gavoille and Pierre Fraigniaud, editors, {\em Proceedings of
  the 30th Annual {ACM} Symposium on Principles of Distributed Computing,
  {PODC} 2011, San Jose, CA, USA, June 6-8, 2011}, pages 187--196. {ACM}, 2011.

\bibitem{puzzle2}
Cynthia Dwork and Moni Naor.
\newblock Pricing via processing or combatting junk mail.
\newblock In Ernest~F. Brickell, editor, {\em Advances in Cryptology ---
  CRYPTO' 92}, pages 139--147, Berlin, Heidelberg, 1993. Springer Berlin
  Heidelberg.

\bibitem{selfish}
Ittay Eyal and Emin~G{\"u}n Sirer.
\newblock Majority is not enough: Bitcoin mining is vulnerable.
\newblock In Nicolas Christin and Reihaneh Safavi-Naini, editors, {\em
  Financial Cryptography and Data Security}, pages 436--454, Berlin,
  Heidelberg, 2014. Springer Berlin Heidelberg.

\bibitem{cheap2}
Joseph Farrell.
\newblock Cheap talk, coordination, and entry.
\newblock {\em The RAND Journal of Economics}, 18(1):34--39, 1987.

\bibitem{koutsoupias}
Amos Fiat, Anna Karlin, Elias Koutsoupias, and Christos Papadimitriou.
\newblock Energy equilibria in proof-of-work mining.
\newblock In {\em Proceedings of the 2019 ACM Conference on Economics and
  Computation}, EC '19, pages 489--502, New York, NY, USA, 2019. ACM.

\bibitem{DBLP:conf/tcc/FuchsbauerKN10}
Georg Fuchsbauer, Jonathan Katz, and David Naccache.
\newblock Efficient rational secret sharing in standard communication networks.
\newblock In Daniele Micciancio, editor, {\em Theory of Cryptography, 7th
  Theory of Cryptography Conference, {TCC} 2010, Zurich, Switzerland, February
  9-11, 2010. Proceedings}, volume 5978 of {\em Lecture Notes in Computer
  Science}, pages 419--436. Springer, 2010.

\bibitem{rationalframework}
Juan Garay, Jonathan Katz, Ueli Maurer, Bj\"{o}rn Tackmann, and Vassilis Zikas.
\newblock Rational protocol design: Cryptography against incentive-driven
  adversaries.
\newblock In {\em Proceedings of the 2013 IEEE 54th Annual Symposium on
  Foundations of Computer Science}, FOCS '13, pages 648--657, Washington, DC,
  USA, 2013. IEEE Computer Society.

\bibitem{backbone}
Juan Garay, Aggelos Kiayias, and Nikos Leonardos.
\newblock The bitcoin backbone protocol: Analysis and applications.
\newblock In Elisabeth Oswald and Marc Fischlin, editors, {\em Advances in
  Cryptology - EUROCRYPT 2015}, pages 281--310, Berlin, Heidelberg, 2015.
  Springer Berlin Heidelberg.

\bibitem{cheap3}
Dino Gerardi.
\newblock Unmediated communication in games with complete and incomplete
  information.
\newblock {\em Journal of Economic Theory}, 114(1):104 -- 131, 2004.

\bibitem{perf}
Arthur Gervais, Ghassan~O. Karame, Karl W\"{u}st, Vasileios Glykantzis, Hubert
  Ritzdorf, and Srdjan Capkun.
\newblock On the security and performance of proof of work blockchains.
\newblock In {\em Proceedings of the 2016 ACM SIGSAC Conference on Computer and
  Communications Security}, CCS '16, pages 3--16, New York, NY, USA, 2016. ACM.

\bibitem{faircomputation}
Adam Groce and Jonathan Katz.
\newblock Fair computation with rational players.
\newblock In David Pointcheval and Thomas Johansson, editors, {\em Advances in
  Cryptology -- EUROCRYPT 2012}, pages 81--98, Berlin, Heidelberg, 2012.
  Springer Berlin Heidelberg.

\bibitem{10.1007/978-3-642-31585-5_50}
Adam Groce, Jonathan Katz, Aishwarya Thiruvengadam, and Vassilis Zikas.
\newblock Byzantine agreement with a rational adversary.
\newblock In Artur Czumaj, Kurt Mehlhorn, Andrew Pitts, and Roger Wattenhofer,
  editors, {\em Automata, Languages, and Programming}, pages 561--572, Berlin,
  Heidelberg, 2012. Springer Berlin Heidelberg.

\bibitem{DBLP:journals/corr/abs-1805-08281}
Cyril Grunspan and Ricardo P{\'{e}}rez{-}Marco.
\newblock On profitability of selfish mining.
\newblock {\em CoRR}, abs/1805.08281, 2018.

\bibitem{users}
{\"O}nder G{\"u}rcan, Antonella Del~Pozzo, and Sara Tucci-Piergiovanni.
\newblock On the bitcoin limitations to deliver fairness to users.
\newblock In Herv{\'e} Panetto, Christophe Debruyne, Walid Gaaloul, Mike
  Papazoglou, Adrian Paschke, Claudio~Agostino Ardagna, and Robert Meersman,
  editors, {\em On the Move to Meaningful Internet Systems. OTM 2017
  Conferences}, pages 589--606, Cham, 2017. Springer International Publishing.

\bibitem{Halpern:2016:RCE:2933057.2933088}
Joseph~Y. Halpern and Xavier Vila\c{c}a.
\newblock Rational consensus: Extended abstract.
\newblock In {\em Proceedings of the 2016 ACM Symposium on Principles of
  Distributed Computing}, PODC '16, pages 137--146, New York, NY, USA, 2016.
  ACM.

\bibitem{brit}
Jr. Harvey S.~James.
\newblock {\em Incentive compatibility}.
\newblock Encyclopedia Britannica, inc., 4 2014.
\newblock Encyclopedia Britannica
  \url{https://www.britannica.com/topic/incentive-compatibility}.

\bibitem{squir}
Charlie Hou, Mingxun Zhou, Yansquir Ji, Phil Daian, Florian Tramer, Giulia
  Fanti, and Ari Juels.
\newblock Squirrl: Automating attack discovery on blockchain incentive
  mechanisms with deep reinforcement learning, 2019.

\bibitem{kroll}
Edward W~Felten Joshua A~Kroll, Ian C~Davey.
\newblock The economics of bitcoin mining, or bitcoin in the presence of
  adversaries.
\newblock In {\em Proceedings of WEIS}, 2013.

\bibitem{puzzle3}
Ari Juels and John~G. Brainard.
\newblock Client puzzles: {A} cryptographic countermeasure against connection
  depletion attacks.
\newblock In {\em {NDSS}}. The Internet Society, 1999.

\bibitem{gamecrypto}
Jonathan Katz.
\newblock Bridging game theory and cryptography: Recent results and future
  directions.
\newblock In Ran Canetti, editor, {\em Theory of Cryptography}, pages 251--272,
  Berlin, Heidelberg, 2008. Springer Berlin Heidelberg.

\bibitem{UC2}
Jonathan Katz, Ueli Maurer, Bj{\"o}rn Tackmann, and Vassilis Zikas.
\newblock Universally composable synchronous computation.
\newblock In Amit Sahai, editor, {\em Theory of Cryptography}, pages 477--498,
  Berlin, Heidelberg, 2013. Springer Berlin Heidelberg.

\bibitem{essential}
Yoav~Shoham Kevin Leyton-Brown.
\newblock {\em Essentials of Game Theory: A Concise Multidisciplinary
  Introduction (Synthesis Lectures on Artificial Intelligence and Machine
  Learning)}.
\newblock Morgan \& Claypool Publishers, 2008.

\bibitem{mininggame}
Aggelos Kiayias, Elias Koutsoupias, Maria Kyropoulou, and Yiannis Tselekounis.
\newblock Blockchain mining games.
\newblock In {\em Proceedings of the 2016 ACM Conference on Economics and
  Computation}, EC '16, pages 365--382, New York, NY, USA, 2016. ACM.

\bibitem{ouroboros}
Aggelos Kiayias, Alexander Russell, Bernardo David, and Roman Oliynykov.
\newblock Ouroboros: A provably secure proof-of-stake blockchain protocol.
\newblock In Jonathan Katz and Hovav Shacham, editors, {\em Advances in
  Cryptology -- CRYPTO 2017}, pages 357--388, Cham, 2017. Springer
  International Publishing.

\bibitem{absolutecost}
Abhiram Kothapalli, Andrew Miller, and Nikita Borisov.
\newblock Smartcast: An incentive compatible consensus protocol using smart
  contracts.
\newblock In Michael Brenner, Kurt Rohloff, Joseph Bonneau, Andrew Miller,
  Peter~Y.A. Ryan, Vanessa Teague, Andrea Bracciali, Massimiliano Sala,
  Federico Pintore, and Markus Jakobsson, editors, {\em Financial Cryptography
  and Data Security}, pages 536--552, Cham, 2017. Springer International
  Publishing.

\bibitem{DBLP:journals/toplas/LamportSP82}
Leslie Lamport, Robert~E. Shostak, and Marshall~C. Pease.
\newblock The byzantine generals problem.
\newblock {\em {ACM} Trans. Program. Lang. Syst.}, 4(3):382--401, 1982.

\bibitem{leonardos2019presto}
Stefanos Leonardos, Daniel Reijsbergen, and Georgios Piliouras.
\newblock Presto: A systematic framework for blockchain consensus protocols,
  2019.

\bibitem{DBLP:conf/podc/LepinskiMP04}
Matt Lepinski, Silvio Micali, Chris Peikert, and Abhi Shelat.
\newblock Completely fair {SFE} and coalition-safe cheap talk.
\newblock In Soma Chaudhuri and Shay Kutten, editors, {\em Proceedings of the
  Twenty-Third Annual {ACM} Symposium on Principles of Distributed Computing,
  {PODC} 2004, St. John's, Newfoundland, Canada, July 25-28, 2004}, pages
  1--10. {ACM}, 2004.

\bibitem{coalitionsafecheap}
Matt Lepinski, Silvio Micali, Chris Peikert, and Abhi Shelat.
\newblock Completely fair sfe and coalition-safe cheap talk.
\newblock In {\em Proceedings of the Twenty-third Annual ACM Symposium on
  Principles of Distributed Computing}, PODC '04, pages 1--10, New York, NY,
  USA, 2004. ACM.

\bibitem{inclusive}
Yoad Lewenberg, Yonatan Sompolinsky, and Aviv Zohar.
\newblock Inclusive block chain protocols.
\newblock In Rainer B{\"o}hme and Tatsuaki Okamoto, editors, {\em Financial
  Cryptography and Data Security}, pages 528--547, Berlin, Heidelberg, 2015.
  Springer Berlin Heidelberg.

\bibitem{whale}
Kevin Liao and Jonathan Katz.
\newblock Incentivizing blockchain forks via whale transactions.
\newblock In Michael Brenner, Kurt Rohloff, Joseph Bonneau, Andrew Miller,
  Peter~Y.A. Ryan, Vanessa Teague, Andrea Bracciali, Massimiliano Sala,
  Federico Pintore, and Markus Jakobsson, editors, {\em Financial Cryptography
  and Data Security}, pages 264--279, Cham, 2017. Springer International
  Publishing.

\bibitem{fairness}
J.~{Liu}, W.~{Li}, G.~O. {Karame}, and N.~{Asokan}.
\newblock Toward fairness of cryptocurrency payments.
\newblock {\em IEEE Security Privacy}, 16(3):81--89, May 2018.

\bibitem{survey}
Ziyao Liu, Nguyen~Cong Luong, Wenbo Wang, Dusit Niyato, Ping Wang, Ying-Chang
  Liang, and Dong~In Kim.
\newblock A survey on applications of game theory in blockchain, 2019.

\bibitem{demystifying}
Loi Luu, Jason Teutsch, Raghav Kulkarni, and Prateek Saxena.
\newblock Demystifying incentives in the consensus computer.
\newblock In {\em Proceedings of the 22Nd ACM SIGSAC Conference on Computer and
  Communications Security}, CCS '15, pages 706--719, New York, NY, USA, 2015.
  ACM.

\bibitem{10.1007/11818175_11}
Anna Lysyanskaya and Nikos Triandopoulos.
\newblock Rationality and adversarial behavior in multi-party computation.
\newblock In Cynthia Dwork, editor, {\em Advances in Cryptology - CRYPTO 2006},
  pages 180--197, Berlin, Heidelberg, 2006. Springer Berlin Heidelberg.

\bibitem{anonymity1}
Sarah Meiklejohn, Marjori Pomarole, Grant Jordan, Kirill Levchenko, Damon
  McCoy, Geoffrey~M. Voelker, and Stefan Savage.
\newblock A fistful of bitcoins: Characterizing payments among men with no
  names.
\newblock In {\em Proceedings of the 2013 Conference on Internet Measurement
  Conference}, IMC '13, pages 127--140, New York, NY, USA, 2013. ACM.

\bibitem{nakamoto}
Satoshi Nakamoto.
\newblock {\em Bitcoin: A Peer-to-Peer Electronic Cash System}, 2008.
\newblock \url{http://bitcoin.org/bitcoin.pdf}.

\bibitem{selfish4}
K.~{Nayak}, S.~{Kumar}, A.~{Miller}, and E.~{Shi}.
\newblock Stubborn mining: Generalizing selfish mining and combining with an
  eclipse attack.
\newblock In {\em 2016 IEEE European Symposium on Security and Privacy (EuroS
  P)}, pages 305--320, March 2016.

\bibitem{costlycom}
Rafael Pass and Joe Halpern.
\newblock Game theory with costly computation: Formulation and application to
  protocol security.
\newblock In {\em Proceedings of the Behavioral and Quantitative Game Theory:
  Conference on Future Directions}, BQGT '10, pages 89:1--89:1, New York, NY,
  USA, 2010. ACM.

\bibitem{rafael}
Rafael Pass, Lior Seeman, and Abhi Shelat.
\newblock Analysis of the blockchain protocol in asynchronous networks.
\newblock In Jean-S{\'e}bastien Coron and Jesper~Buus Nielsen, editors, {\em
  Advances in Cryptology -- EUROCRYPT 2017}, pages 643--673, Cham, 2017.
  Springer International Publishing.

\bibitem{fruitchain}
Rafael Pass and Elaine Shi.
\newblock Fruitchains: A fair blockchain.
\newblock In {\em Proceedings of the ACM Symposium on Principles of Distributed
  Computing}, PODC '17, pages 315--324, New York, NY, USA, 2017. ACM.

\bibitem{puzzle4}
R.~L. Rivest, A.~Shamir, and D.~A. Wagner.
\newblock Time-lock puzzles and timed-release crypto.
\newblock Technical report, Cambridge, MA, USA, 1996.

\bibitem{anonymity2}
Dorit Ron and Adi Shamir.
\newblock Quantitative analysis of the full bitcoin transaction graph.
\newblock In Ahmad-Reza Sadeghi, editor, {\em Financial Cryptography and Data
  Security}, pages 6--24, Berlin, Heidelberg, 2013. Springer Berlin Heidelberg.

\bibitem{selfish2}
Ayelet Sapirshtein, Yonatan Sompolinsky, and Aviv Zohar.
\newblock Optimal selfish mining strategies in bitcoin.
\newblock In Jens Grossklags and Bart Preneel, editors, {\em Financial
  Cryptography and Data Security}, pages 515--532, Berlin, Heidelberg, 2017.
  Springer Berlin Heidelberg.

\bibitem{DBLP:journals/csur/Schneider90}
Fred~B. Schneider.
\newblock Implementing fault-tolerant services using the state machine
  approach: {A} tutorial.
\newblock {\em {ACM} Comput. Surv.}, 22(4):299--319, 1990.

\bibitem{gap}
Itay Tsabary and Ittay Eyal.
\newblock The gap game.
\newblock In {\em Proceedings of the 2018 ACM SIGSAC Conference on Computer and
  Communications Security}, CCS '18, pages 713--728, New York, NY, USA, 2018.
  ACM.

\bibitem{DBLP:journals/geb/UrbanoV04}
Amparo Urbano and Jos{\'{e}}~Enrique Vila.
\newblock Unmediated communication in repeated games with imperfect monitoring.
\newblock {\em Games and Economic Behavior}, 46(1):143--173, 2004.

\bibitem{10.1007/978-3-319-02786-9_14}
John~Ross Wallrabenstein and Chris Clifton.
\newblock Equilibrium concepts for rational multiparty computation.
\newblock In Sajal~K. Das, Cristina Nita-Rotaru, and Murat Kantarcioglu,
  editors, {\em Decision and Game Theory for Security}, pages 226--245, Cham,
  2013. Springer International Publishing.

\end{thebibliography}
